\pdfoutput=1

\documentclass[journal]{IEEEtran}   

\usepackage{multirow}
\usepackage{diagbox}
\usepackage{amssymb}
\usepackage{amsmath}
\usepackage{graphicx}
\usepackage{cite}
\usepackage{citesort}
\usepackage{booktabs}
\usepackage{subfigure}
\usepackage{graphicx,epstopdf}
\usepackage{epsfig}	
\usepackage{cite,graphicx,amsmath,amssymb}
\usepackage{comment}
\usepackage{amssymb}
\usepackage{amsmath}
\usepackage{cite}
\usepackage{url}
\usepackage{color}
\usepackage{xcolor}
\usepackage{cite,graphicx,amsmath,amssymb}
\usepackage{subfigure}
\usepackage{citesort}
\usepackage{fancyhdr}
\usepackage{mdwmath}
\usepackage{mdwtab}
\usepackage{caption}
\usepackage{amsthm}
\usepackage{lipsum}
\usepackage{array}
\usepackage{booktabs}
\usepackage{colortbl}
\newtheorem{theorem}{Theorem}

\newtheorem{lemma}{Lemma}

\newtheorem{corollary}{Corollary}

\newtheorem{remark}{Remark}  

\makeatletter
\def\ScaleIfNeeded{
\ifdim\Gin@nat@width>\linewidth \linewidth \else \Gin@nat@width
\fi } \makeatother

\begin{document}

\title{\Huge{Secure Communication of Active RIS Assisted NOMA Networks}}

\author{Xuehua\ Li,~\IEEEmembership{Member,~IEEE}, Yingjie~Pei, Xinwei\ Yue,~\IEEEmembership{Senior Member,~IEEE}, Yuanwei\ Liu,~\IEEEmembership{Senior Member~IEEE}, and  Zhiguo~Ding,~\IEEEmembership{Fellow, IEEE}

\thanks{This work was supported in part by National Natural Science Foundation of China under Grant 62071052, and in part by the Beijing Natural Science Foundation under Grant L222004. \emph{(Corresponding author: Xinwei Yue.)}}
\thanks{X. Li and X. Yue are with the Key Laboratory of Information and Communication Systems, Ministry of Information Industry and also with the Key Laboratory of Modern Measurement $\&$ Control Technology, Ministry of Education, Beijing Information Science and Technology University, Beijing 100101, China (email: \{lixuehua and xinwei.yue\}@bistu.edu.cn).}
\thanks{Y. Pei is with the School of Information and Communication Engineering, Beijing University of Posts and Telecommunications, Beijing 100876, China (email: yingjie.pei@bupt.edu.cn).}
\thanks{Y. Liu is with the School of Electronic Engineering and Computer Science, Queen Mary University of London, London E1 4NS, U.K. (email: yuanwei.liu@qmul.ac.uk).}
\thanks{Z. Ding is with Department of Electrical Engineering and Computer Science, Khalifa University, Abu Dhabi, UAE. (e-mail: zhiguo.ding@gmail.com).}
}

\maketitle

\begin{abstract}
  As a revolutionary technology, reconfigurable intelligent surface (RIS) has been deemed as an indispensable part of the 6th generation communications due to its inherent ability to regulate the wireless channels. However, passive RIS (PRIS) still suffers from some pressing issues, one of which is that the fading of the entire reflection link is proportional to the product of the distances from the base station to the PRIS and from the PRIS to the users, i.e., the productive attenuation. To tackle this problem, active RIS (ARIS) has been proposed to reconfigure the wireless propagation condition and alleviate the productive attenuation. In this paper, we investigate the physical layer security of the ARIS assisted non-orthogonal multiple access (NOMA) networks with the attendance of external and internal eavesdroppers. To be specific, the closed-form expressions of secrecy outage probability (SOP) and secrecy system throughput are derived by invoking both imperfect successive interference cancellation (ipSIC) and perfect SIC. The secrecy diversity orders of legitimate users are obtained at high signal-to-noise ratios. Numerical results are presented to verify the accuracy of the theoretical expressions and indicate that: i) The SOP of ARIS assisted NOMA networks exceeds that of PRIS-NOMA, ARIS/PRIS-assisted orthogonal multiple access (OMA); ii) Due to the balance between the thermal noise and residual interference, introducing excess reconfigurable elements at ARIS is not helpful to reduce the SOP; and iii) The secrecy throughput performance of ARIS-NOMA networks outperforms that of PRIS-NOMA and ARIS/PRIS-OMA networks.
\end{abstract}
\section{Introduction}
\makeatletter 
\let\myorg@bibitem\bibitem
\def\bibitem#1#2\par{%
  \@ifundefined{bibitem@#1}{%
    \myorg@bibitem{#1}#2\par
  }{
    \begingroup
      \color{\csname bibitem@#1\endcsname}%
      \myorg@bibitem{#1}#2\par
    \endgroup
  }
}

\makeatother 

The considerable demand for high data capacity and low transmission latency of wireless communication has experienced a remarkable proliferation in the past decades \cite{you2021towards,Saad2020vision6G}. To alleviate the pressure brought by the frequent and enormous information exchange, non-orthogonal multiple access (NOMA) has become a hotspot in both academia and industry with its superior spectrum efficiency, massive connection and user fairness \cite{2017YuanweiNOMA,2018XiaomingPDMA,2022XinyueNGMA}. Generally, the primary characteristic of NOMA is that signals are assigned with various power and superimposed together within the same time/frequency/code resource blocks before transmitting \cite{2017Zhiguo5GNOMA,2016ZhiguouserpairNOMA}. In this case, users/servers at the edge of cells are likely to received higher signal-to-noise ratio (SNR) gains and achieve satisfying quality of service \cite{2018XinweiUnifiedNOMA,DehuanWanNakagamiNOMA}. However, NOMA with advanced spectral efficiency cannot guarantee the physical layer security of communication networks since the radio signals are in danger of being individually overheard by passive eavesdroppers (Eve) due to the open characteristics of wireless channels\footnote{Different from the active Eve which can continuously wiretap users' information by contaminating the estimated channel
state information (CSI) and confusing the standard beamforming via active attacks like pilot spoofing attack \cite{2021XiaomingPilotAttackRIS,2021XiliangRIS}, the passive Eve is suppressed and keeps overhearing legitimate signals without actively exchanging information with other nodes. Therefore, the CSI of passive Eve is difficult to acquire since it always remain silent.}, which is known as the external eavesdropping scenario  \cite{20206G,Mag}. Hence, it is critical to evaluate the physical layer security of NOMA networks. The authors of \cite{2018BeixiongTwoWayNOMA} introduced artificial noise to confuse the vicious Eves and achieve secure NOMA communications. Considering the inaccuracy of artificial noise when CSI is unavailable, the authors in \cite{2018DerrickSecureNOMA} utilized the interference between users to perturb the reception at Eves. Further, the authors of \cite{2020XinweiSecureNOMA} analyzed the secrecy performance of users with the consideration of imperfect successive interference cancellation (ipSIC) and perfect SIC (pSIC). Apart from the external eavesdropping scenarios, users are able to act as internal Eves to wiretap others' information from the superimposed signals, which can cause more covert safety hazard since the internal Eves are one-way transparent to the base station (BS) \cite{2020JiafanInternalEve}. The above-mentioned methods can improve the secrecy of NOMA networks at the expense of uncontrollability and high hardware complexity, so there is an urgent need for a low-cost, easy-to-deploy and flexible technology for the private transmission of NOMA networks.

Reconfigurable intelligent surface (RIS) has been recognized as a prospective technology for the next-generation wireless communication \cite{2020RenzoRIS,Basharat2021RIS}. Generally, passive RIS (PRIS) is composed of abundant low-cost reconfigurable elements which can independently adjust the phase shifts of the incident waves and thus realize the flexible control of signals' propagation directions \cite{2021QingqingRISTutorial,2020XidongRISNOMAopt,2022yuanweiMultiUserRISNOMA}. The authors of \cite{2020BeixiongRISOFDM} studied the application of PRIS to wireless networks, which attains the superior achievable rate. In \cite{2020TaoqinRIS}, the authors investigated the deployment of PRIS to realize enhanced ergodic capacity and outage behaviours. The authors of \cite{2020ZhiguoIRSNOMA} considered the influence of PRIS when the reconfigurable emelments are set to coherent/random phase shifts. 
In addition, physical layer security comes to a situation where opportunities and challenges coexist due to the appearance of PRIS \cite{2020QingqingMag}. On one hand, by intelligently adjusting the phase shifts of reconfigurable elements, the power of received signals at legitimate users can be boosted so as to ensure secrecy transmission \cite{2019ZhangRuiPLSRIS,2019XianghaoYuPLSRIS}. On the other hand, PRIS can also act as a malicious jammer with the purpose of making the reflection signal to be destructively superimposed on the direct signal, thereby compromising the SINR of the legitimate user \cite{2020BinLyuEavesRIS}. An artificial noise aided multiple input multiple output (MIMO) sercecy communication framework with eavesdropping PRIS was proposed in \cite{2021AlexandropoulosMIMOEaveRIS,2023AlexandropoulosMIMOEaveRIS}, where nonzero secrecy rates can be guaranteed with/without legitimate PRIS. Given the efficient adaptability between NOMA and PRIS, the authors of \cite{2020TianweiRISNOMA} surveyed the spectral and energy efficiency of PRIS-NOMA networks. As a further advance, the authors of \cite{YuanweiRISNOMAUpDolink} evaluated the outage probability and ergodic rate of PRIS-NOMA networks. In addition, the authors of \cite{2021XidongIRSReAllo} researched the capacity and rate regions of PRIS-NOMA networks with multiple users were studied in detail. Regarding to multiple users, the superiority of PRIS-NOMA networks with multiple ordered users was revealed in \cite{2022XinweiYueRISNOMA}, where the better outage and ergodic performance were achieved than traditional relaying schemes. The authors of \cite{2022ZiweiLiuTWRISNOMA} highlighted the enhancement in two-way NOMA networks by the virtue of PRIS. 

In light of the aforementioned discussions, topics about the secrecy performance of PRIS-NOMA networks enter the consciousness of academia. More especially, the physical layer security of PRIS-NOMA networks was researched in \cite{2020YangliangRISPLS}, where the secrecy outage probability (SOP) expressions of legitimate users were derived. In \cite{2022CaihongPLSRISNOMA,2023YingjiePLSPRISNOMA}, 1-bit coding scheme was utilized as an feasible method to construct the phase shifts to guarantee the secrecy transmission of PRIS-NOMA networks. The authors of \cite{2021QingqingPLSRISNOMA} proposed an artificial noise based optimization strategy to realize the secure communication of PRIS-NOMA networks. Artificial jamming was also harnessed in \cite{2022HanhuiPLSRISNOMA} to guarantee the secure PRIS-NOMA transmission by applying  within both external and internal wiretapping scenarios. In \cite{2022WangweiPLSRISNOMA}, the authors investigated the maximization of sum rate by designing the transmitting beamforming at BS as well as the reflecting vector at PRIS. Further, the authors of \cite{2022zhangzhengPLSRISNOMA} verified the secrecy superiority of multiple distributed PRISs aided NOMA networks with the conditional of equal allocation of reflecting components. The SOP and average secrecy capacity of PRIS-NOMA networks were analyzed in \cite{2022TianweiPLSNOMARIS} over Nakagami-\emph{m} channels. Recently, the authors of \cite{2023ZhangqianRISNOMAHI} investigated the secure transmission of PRIS-NOMA networks by taking into account the hardware impairments at transceiver. A novel two-way training scheme was proposed in \cite{2023BaiTongRISNOMAPLS} to prevent pilot spoofing attack in PRIS-NOMA networks. In addition, the secure communication of dual-functional RIS-aided NOMA with imperfect CSI was surveyed in \cite{2022WangWenRISNOMAPLS}, where secrecy energy efficiency was maximized by jointly optimizing power allocation and active/passive beamforming.

However, due to the existence of product path loss, the negative impact of multiplicative attenuation on cascaded wireless links can not be neglected when only phase shift regulation is employed \cite{2020BjornsonRIS}. To resist the severe channel fading, active RIS (ARIS) is proposed recently to magnify the incident signals by introducing extra reflection-type amplifiers into traditional PRIS circuits \cite{2022ChongwenRIS}. Although this design architecture is similar to the conventional full-duplex (FD) amplify-and-forward (AF) relaying, the FD-AF requires two full time slots to successfully transmit one symbol, whereas the ARIS can finish the process in a single time slot \cite{ntontin2019multi}. In addition, without the effects of multiplicative fading, ARIS can reap significant performance gains with a small number of reconfiguration elements, resulting in remarkable hardware cost savings compared to MIMO \cite{2021KhoshafaARIS}. The authors of \cite{zhang2021active} verified that the total power consumed by ARIS is lower than that of its passive counterpart to achieve the same performance, highlighting the huge potential of ARIS. 
As a further advance, the authors of \cite{2021ChangshengYouARISorPRIS} showed that the best deployment position of ARIS in uplink/downlink was as close to the target receiver as possible. In \cite{CunhuaPanARIS}, two power allocation strategies of ARIS assisted wireless networks were investigated and the performance of ARIS exceeded that of PRIS with sufficient power budget. With the goal of reducing the additional power cost, a hardware structure including a pair of back-to-back placed PRISs with one power amplifier was proposed in \cite{2022TasciRIS}, where the signals received by PRIS-1 is processed via a single amplifier, and then reflected by PRIS-2. For the same purpose, the authors of \cite{2022LinglongSunARIS} further reduce hardware and energy overheads by considering that multiple elements control their phase shifts independently but share a same power amplifier. In terms of secure transmission, the authors in \cite{2022LimengDongARIS} focused on maximizing the secrecy rate with properly designing the reflection matrix at ARIS and beamforming vector at BS.

\subsection{Motivations and Contributions}
The existing treatises have provided a solid foundation in the aspects of PRIS-NOMA networks and their secrecy transmission strategies \cite{2022zhangzhengPLSRISNOMA,2022CaihongPLSRISNOMA}. Little is known, however, about the physical layer security of ARIS-NOMA networks, where the ARIS is regarded as an effective solution to alleviate the productive fading and enhance the reliability of radio channels.  The authors of \cite{2022TianweiPLSNOMARIS} derived the SOP and secrecy capacity expressions to estimate secrecy performance of PRIS-NOMA networks, where the phase of cascaded PRIS channels are coherently well matched. However, the coherent phase shift scheme may cause excessive signalling overhead and is impractical because of the finite resolution of phase shifters at PRIS \cite{ZhiguoSimpleDesign,2020ZhiguoIRSNOMA}. Therefore, the on-off control is adopted as a viable scheme to redesign the phase shifts of ARIS for realizing the secure transmission of ARIS-NOMA networks. In \cite{2022LimengDongARIS}, the authors jointly redesigned the beamforming and reflection coefficients at ARIS to protect the privacy of signals, but only OMA transmission schemes were taken into account. As far as we know, the physical layer security of ARIS-NOMA networks haven't been investigated yet and some critical issues need to be further explored. On the one hand, does the ARIS-NOMA network exhibit security advantages over the PRIS-NOMA and traditional OMA transmission schemes, given the same system power budgets? On the other hand, what is the relationship between the residual interference caused by the ipSIC and the thermal noise generated at ARIS, and how do they affect the secure communication in ARIS-NOMA networks? In this case, we specifically consider the secrecy performance of ARIS-NOMA networks where the superimposed signals are sent to a paired users by the virtue of an ARIS. The adverse influences brought by external and internal Eves are both discussed in detail. Moreover, the closed-form expressions of two pivotal indicators, i.e., secrecy outage behaviours and secrecy throughput, are both obtained by taking into account ipSIC and pSIC. According to the proposed schemes, the main contributions of this manuscript can be summarized as follows:
\begin{enumerate}
  \item We propose an ARIS-NOMA secure communication framework with the presence of external and internal Eves who are attempting to overhear the information of legitimate users. Given the finite resolution of ARIS, a typical random phase shifting design called on-off control scheme is harnessed to dispose the phase shifts. On this basis, the SOP is selected as a key indicator to evaluate the secrecy performance. Furthermore, we also survey the superiority of ARIS-NOMA networks in the secrecy throughput under delay-limited transmission mode.
  \item For the external eavesdropping scenario, we investigate the physical layer security of ARIS-NOMA networks. To be specific, we derive the closed-form and asymptotic expressions of SOP for legitimate users by taking into account ipSIC and pSIC. The proposed ARIS-NOMA networks are capable of achieving the superior secrecy outage behaviours compared to PRIS-NOMA, PRIS-OMA and several conventional relay transmissions involving AF and decode-and-forward (DF) relaying under the condition of the same power budget. We further analyze the secrecy diversity order of each user, which is largely affected by the residential interference from ipSIC and results in the appearance of error floor.
  \item For the internal eavesdropping scenario, we derive the closed-form and asymptotic expressions of SOP for the far user to wiretap near user's signal with the consideration of ipSIC and pSIC in ARIS-NOMA networks. To reap more insights, the secrecy diversity order of the near user is further obtained, which equals \emph{zero}/\emph{one} for ipSIC and pSIC, respectively. We reveal that prominent improvement can be achieved by dividing more power to the near user since it is helpful to enhance its received signal-to-interference-plus-noise ratio (SINR) and suppress the wiretapping ability of the internal Eve.
  \item Numerical results show the progressiveness of the ARIS-NOMA networks and confirm the accuracy of the theoretical analyses. In addition, three valuable insights can be observed: 1) the arbitrary increase in the number of active components and amplification factor is not conducive to enhancing the security of the ARIS-NOMA networks due to the existence of thermal noise; 2) both thermal noise and residual interference can compromise the secure transmission of the ARIS-NOMA networks and there is a mutual constraint between them; and 3) under the condition of ipSIC, the residential interference is more detrimental to the secrecy throughput of PRIS-NOMA than ARIS-NOMA networks, which presents that the latter features advanced robustness.
\end{enumerate}
\subsection{Organization and Notations}
The remainder of this paper is given as follows. The network model and on-off control scheme are proposed in Section \ref{SectionII}. The detailed evaluation of secrecy outage behaviour and secrecy throughput are presented in Section \ref{SectionIII}. Numerical results and analyses are given in Section \ref{SectionIV} followed by conclusions provided in Section \ref{SectionV}.

The primary notations utilized in this paper are explained as follows. $\mathbb{E}\left\{  \cdot  \right\}$ indicates the expectation operation. The cumulative distribution function (CDF) and the probability density function (PDF) with parameter \emph{X} are given by ${F_X}\left(  \cdot  \right)$ and ${f_X}\left(  \cdot  \right)$, respectively. ${{\mathbf{I}}_p }$ refers to a $p  \times p $ identity matrix and ${{\mathbf{1}}_q }$ denotes a $q  \times 1$ all-ones column vector. $ \otimes $ represents Kronecker product.

\section{Network Model}\label{SectionII}
\subsection{Network Descriptions}
Considering an ARIS-assisted NOMA secure communication scenario as illustrated in Fig. 1, where the superposed signal is transmitted from a BS to two non-orthogonal users, i.e., user \emph{f} and user \emph{n}, via the assistance of ARIS in the presence of an Eve. To provide the straightforward analyses, we suppose that BS and users are both equipped with single antenna. Different from the sub-connected ARIS networks displayed in \cite{2022TasciRIS,2022LinglongSunARIS}, we suppose that each ARIS element owns a dedicated active reflection-type amplifier like current-inverting converters \cite{lonvcar2019ultrathin}, asymmetric current mirrors \cite{2012Geoffrey4GHz}, or some integrated circuits \cite{2012KishorAntenna}, for achieving excellent performance. The \emph{M} active reconfigurable elements at ARIS are able to amplify the incident signal to resist the bottleneck of multiplicative fading \cite{CunhuaPanARIS}. The complex channel gain from BS to ARIS and from ARIS to user $\varphi $ and Eve are denoted as ${{\mathbf{h}}_{br}} \in {\mathbb{C}^{M \times 1}}$, ${{\mathbf{h}}_{r\varphi}} \in {\mathbb{C}^{M \times 1}}$ and ${{\mathbf{h}}_{re}} \in {\mathbb{C}^{M \times 1}}$ with $\varphi  \in \{ f,n\} $, respectively. Considering the radio signals will be blocked by obstacles when propagating in the actual urban scenarios, the direct links between BS and users/Eve are supposed to be entirely obstructed. In this case, the wireless links involved in the ARIS-NOMA networks are modeled as Rayleigh fading channels \footnote{Given the random phase shifting process and urban districts with rich scattering components, Rayleigh fading model is utilized in this paper. In the future we will relax this restriction and consider more realistic assumptions, such as Rician and Nakagami-\emph{m} fading channels.}. We assume the CSI of users can be obtained at ARIS, but not the Eve since it always keeps silent.

\begin{figure}[t!]
    \begin{center}
        \includegraphics[width=2.93in,  height=2.2in]{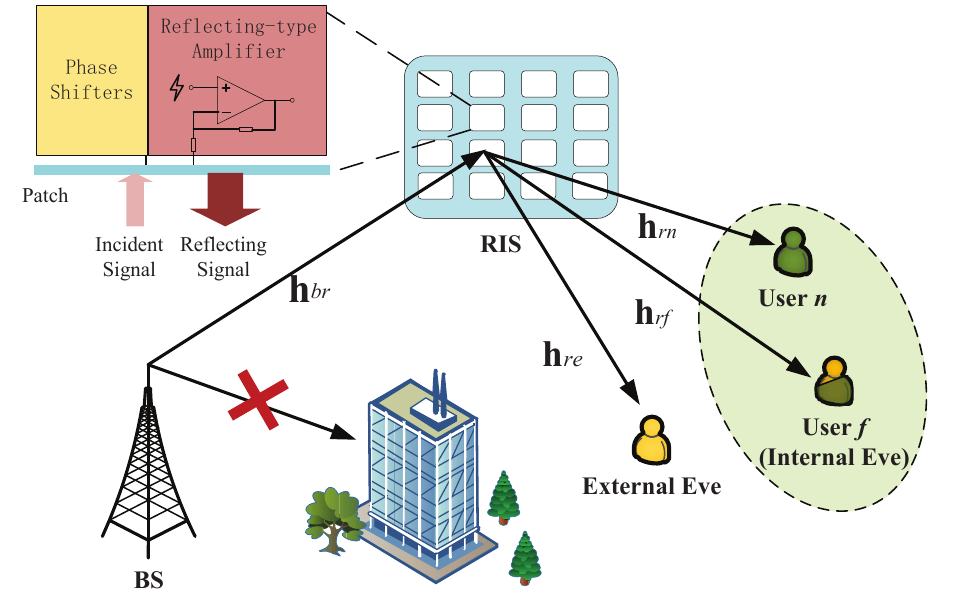}
        \caption{An illustration of ARIS assisted NOMA secure communication networks, where the ARIS consisting of a novel reflecting-type amplifier can simultaneously reflects and amplifies signals to user \emph{n} and user \emph{f} located in the same NOMA cluster.}
        \label{sys_model.eps}
    \end{center}
\end{figure}

\subsection{Signal Model}
In ARIS-NOMA secure communication networks, the superposed signal is transmitted from BS to ARIS, and then to the users after being amplified by the active elements at ARIS. Hence the received signal at user $\varphi$ is shown as
\begin{small}
\begin{align}\label{the received signal at users}
{y_\varphi } = \underbrace {{\mathbf{h}}_{r\varphi }^H{\mathbf{\Theta }}{{\mathbf{h}}_{br}}\left( {\sqrt {P_{BS}^{act}{a_n}} {x_n} + \sqrt {P_{BS}^{act}{a_f}} {x_f}} \right)}_{{\text{Desired signal}}} + \underbrace {{\mathbf{h}}_{r\varphi }^H{\mathbf{\Theta }}{{\mathbf{n}}_t}}_{{\text{Thermal noise}}} + \underbrace {{n_\varphi }}_{{\text{AWGN}}},
\end{align}
\end{small}where the ${P_{BS}^{act}}$ is the transmission power at BS in  ARIS-NOMA networks and ${{x_\varphi}}$ denotes user $\varphi$'s signal with unity power, i.e., $\mathbb{E}\{ {\left| {{x_\varphi}} \right|^2}\}  = 1$. To guarantee user fairness, ${{a_\varphi}}$ represents the power allocation factor of user $\varphi$, which satisfies the relationship ${a_n} + {a_f} = 1$. Denoting ${{\mathbf{h}}_{br}} = {\left[ {h_{br}^1 \cdots h_{br}^m \cdots h_{br}^M} \right]^H}$, where $h_{br}^m = \sqrt {\beta d_{br}^{ - \alpha }} \tilde h_{br}^m$ denotes the complex channel coefficient between BS and the \emph{m}-th active reconfigurable element of ARIS, $\alpha$ indicates the path loss exponent and $\beta$ is the frequency dependent factor, $h_{br}^m \sim \mathcal{C}\mathcal{N}\left( {0,{\Omega_{br}}} \right)$ and $\tilde h_{br}^m \sim \mathcal{C}\mathcal{N}\left( {0,1} \right)$. The ${{\mathbf{h}}_{r\varphi}} = {\left[ {h_{r\varphi}^1 \cdots h_{r\varphi}^m \cdots h_{r\varphi}^M} \right]^H}$, where $h_{r\varphi}^m = \sqrt {\beta d_{r\varphi}^{ - \alpha }} \tilde h_{r\varphi}^m$ denotes the complex channel coefficient between the \emph{m}-th active reconfigurable element of ARIS and user $\varphi$, $h_{r\varphi}^m \sim \mathcal{C}\mathcal{N}\left( {0,{\Omega_{r\varphi}}} \right)$ and $\tilde h_{r\varphi}^m \sim \mathcal{C}\mathcal{N}\left( {0,1} \right)$. The ${d_{br}}$ and ${d_{r\varphi}}$ are the distance between BS and ARIS and between ARIS and user $\varphi$, respectively. The reflection matrix is defined as ${\mathbf{\Theta }} = \kappa diag\{ {e^{j{\theta _1}}}, \cdots ,{e^{j{\theta _m}}}, \cdots ,{e^{j{\theta _M}}}\}  = \kappa {\mathbf{\Phi }}$, where $\theta _m$ denotes the phase shift at the \emph{m}-th reconfigurable element and $\kappa$ is the amplification factor which is always larger than 1 due to the amplifier of ARIS\footnote{The amplification factor of each reflection-type amplifier is set to an identical value $\kappa$ to reduce the configuration overhead, which can also be various for more flexible beamforming design.}. Considering that the ARIS are equipped with additional active components, the thermal noise generated at each reconfigurable element is not negligible, which can be denoted as ${{\mathbf{n}}_t} \sim \mathcal{C}\mathcal{N}\left( {{{\mathbf{0}}_M},\sigma _t^2{{\mathbf{I}}_M}} \right)$ with $\sigma _t^2$ representing the amplification thermal noise power \cite{2021ChangshengYouARISorPRIS,CunhuaPanARIS,zhang2021active}. The ${n_\varphi}$ denotes the additive white Gaussian noise (AWGN) at user $\varphi$ both with mean power parameter ${\sigma ^2}$.

The performance of  ARIS is bound to be superior to that of PRIS without setting limitations on the total system power \cite{CunhuaPanARIS}. For the sake of fairness, the total power budgets are supposed to be the same in both ARIS/PRIS-aided NOMA networks, which can be defined as $P_{tot}^{act} = P_{BS}^{act} + P_{RIS}^{act} + Q\left( {{P_{PS}} + {P_{DC}}} \right)$ and $P_{tot}^{pas} = P_{BS}^{pas} + M{P_{PS}}$, respectively. ${{P_{PS}}}$ denotes the power consumption caused by phase shifters and ${{P_{DC}}}$ represents the direct current biasing power consumed by the amplifier installed in each active reconfigurable element.

According to NOMA protocol, the SIC scheme is introduced at user \emph{n} to first decode another user's information and then decode its own signal. In this case, the received SINR at user \emph{n} can be given by
\begin{align}\label{SINR usern}
{\gamma _n^{ipSIC}} = \frac{{{a_n}P_{BS}^{act}{\kappa ^2}{{\left| {{\mathbf{h}}_{rn}^H{\mathbf{\Phi }}{{\mathbf{h}}_{br}}} \right|}^2}}}{{{\kappa ^2}\sigma _t^2{{\left\| {{\mathbf{h}}_{rn}^H{\mathbf{\Phi }}} \right\|}^2} + \varpi P_{BS}^{act}{{\left| {{h_{ipu}}} \right|}^2} + {\sigma ^2}}},
\end{align}
where $\varpi  \in \left[ {0,1} \right]$   is the residual interference degree for SIC. Specifically, $\varpi  = 0$ and $\varpi  \ne 0$ indicate the switchover between pSIC and ipSIC, respectively. Without loss of generality, the residual interference generated by ipSIC is assumed to be modeled as the Rayleigh fading and the relative complex channel parameter is denoted by ${h_{ipu}} \sim \mathcal{C}\mathcal{N}\left( {0,{\Omega_{ipu}}} \right)$.

As for user \emph{f}, deploying SIC is unnecessary since it has larger power allocation and it can acquire its information directly. As a consequence, the SINR for user \emph{f} to decode its own information is shown as
\begin{align}\label{SINR userf}
{\gamma _f} = \frac{{{a_f}P_{BS}^{act}{\kappa ^2}{{\left| {{\mathbf{h}}_{rf}^H{\mathbf{\Phi }}{{\mathbf{h}}_{br}}} \right|}^2}}}{{{a_n}P_{BS}^{act}{\kappa ^2}{{\left| {{\mathbf{h}}_{rf}^H{\mathbf{\Phi }}{{\mathbf{h}}_{br}}} \right|}^2} + {\kappa ^2}\sigma _t^2{{\left\| {{\mathbf{h}}_{rf}^H{\mathbf{\Phi }}} \right\|}^2} + {\sigma ^2}}}.
\end{align}

Given the complex electromagnetic environment can impose a negative effect on the secure communication of the exposed radio signals, the external and internal wiretapping scenarios are both considered to evaluate the secrecy performance of ARIS-NOMA networks.
\subsubsection{External Wiretapping Scenario}
The broadcasting characteristics of wireless signals provide the external Eve opportunities to wiretap users' information and the received signal at external Eve is shown as
\begin{small}
\begin{align}\label{the received signal at EE}
{y_{EE} } = {\mathbf{h}}_{re }^H{\mathbf{\Phi }}{{\mathbf{h}}_{br}}\left( {\sqrt {P_{BS}^{act}{a_n}} {x_n} + \sqrt {P_{BS}^{act}{a_f}} {x_f}} \right) + {\mathbf{h}}_{re }^H{\mathbf{\Phi }}{{\mathbf{n}}_t} + {n_e},
\end{align}
\end{small}where ${{\mathbf{h}}_{re}} = {\left[ {h_{re}^1 \cdots h_{re}^m \cdots h_{re}^M} \right]^H}$ is the channel from ARIS to external Eve, $h_{re}^m = \sqrt {\beta d_{re}^{ - \alpha }} \tilde h_{re}^m$ denotes the complex channel coefficient between the \emph{m}-th active reconfigurable element of ARIS and Eve, $h_{re}^m \sim \mathcal{C}\mathcal{N}\left( {0,{\Omega_{re}}} \right)$ and $\tilde h_{re}^m \sim \mathcal{C}\mathcal{N}\left( {0,1} \right)$. The $n_e$ denotes the AWGN at Eve. Referring to the analysis above, the external Eve can decode the information of user \emph{n} by applying SIC and that of user \emph{f} directly. Hence, the SINR for external Eve to wiretap the signals of user \emph{n} and user \emph{f} can be given by
\begin{align}\label{the SINR EE n}
{\gamma _{e \to n}^{ipSIC}} = \frac{{{a_n}P_{BS}^{act}{\kappa ^2}{{\left| {{\mathbf{h}}_{re}^H{\mathbf{\Phi }}{{\mathbf{h}}_{br}}} \right|}^2}}}{{{\kappa ^2}\sigma _t^2{{\left\| {{\mathbf{h}}_{re}^H{\mathbf{\Phi }}} \right\|}^2} + \varpi P_{BS}^{act}{{\left| {{h_{ipe}}} \right|}^2} + {\sigma _e}^2}},
\end{align}
and
\begin{align}\label{the SINR EE f}
{\gamma _{e \to f}} = \frac{{{a_f}P_{BS}^{act}{\kappa ^2}{{\left| {{\mathbf{h}}_{re}^H{\mathbf{\Phi }}{{\mathbf{h}}_{br}}} \right|}^2}}}{{{a_n}P_{BS}^{act}{\kappa ^2}{{\left| {{\mathbf{h}}_{re}^H{\mathbf{\Phi }}{{\mathbf{h}}_{br}}} \right|}^2} + {\kappa ^2}\sigma _t^2{{\left\| {{\mathbf{h}}_{re}^H{\mathbf{\Phi }}} \right\|}^2} + {\sigma _e}^2}},
\end{align}
respectively, where ${h_{ipe}} \sim \mathcal{C}\mathcal{N}\left( {0,{\Omega_{ipe}}} \right)$ is the residential interference at Eve.

\subsubsection{Internal Wiretapping Scenario}
In this case, the distant user \emph{f} will be treated as an internal Eve since the channel condition of user \emph{f} is weaker than that of user \emph{n}\footnote{It will be more challenging if the near user is regarded as an internal Eve since it can directly decode the confidential signal of the distant user. One feasible method is to enhance distant user's channel gain with the assistance of beamforming optimization and then switch the order of SIC for the purpose of deteriorating the received SINR at internal Eve \cite{2021ZheZhangInternalEveRISNOMA}. The secrecy performance of ARIS-NOMA networks where the near user act as an attacker will be investigated in our future work.
}. As a result, the received signal at internal Eve can be written as
\begin{small}\emph{}
\begin{align}\label{the received signal at IE}
{y_{IE} } = {\mathbf{h}}_{rf }^H{\mathbf{\Phi }}{{\mathbf{h}}_{br}}\left( {\sqrt {P_{BS}^{act}{a_n}} {x_n} + \sqrt {P_{BS}^{act}{a_f}} {x_f}} \right) + {\mathbf{h}}_{rf }^H{\mathbf{\Phi }}{{\mathbf{n}}_t} + {n_e}.
\end{align}
\end{small}At this moment, the SINR for user \emph{f} to detect user \emph{n}'s information is shown as follows
\begin{align}\label{SINR f decode n}
{\gamma _{f \to n}} = \frac{{{a_n}P_{BS}^{act}{\kappa ^2}{{\left| {{\mathbf{h}}_{rf}^H{\mathbf{\Phi }}{{\mathbf{h}}_{br}}} \right|}^2}}}{{{\kappa ^2}\sigma _t^2{{\left\| {{\mathbf{h}}_{rf}^H{\mathbf{\Phi }}} \right\|}^2} + {\sigma _e}^2}}.
\end{align}

\subsection{ARIS-NOMA Networks with On-off Control Scheme}
Considering the finite resolution and hardware limitations of the actual reconfigurable elements, the on-off control is selected as an appropriate scheme to design the phase shift of elements in ARIS, where each component of ARIS is set to 1 (on) or 0 (off)\cite{2022XinweiYueRISNOMA,ZhiguoSimpleDesign}. Specifically, we assume that \emph{M} = \emph{PQ}, where \emph{P} and \emph{Q} are both positive integers. Defining ${\mathbf{V}} = {{\mathbf{I}}_P} \otimes {{\mathbf{1}}_Q} \in {\mathbb{C}^{M \times P}}$ and ${{\mathbf{v}}_p}$ is a vector denoting the \emph{p}-th column of ${\mathbf{V}}$. On this basis, the aforementioned cascaded channels ${{\mathbf{h}}_{_{r\varphi}}^H{\mathbf{\Phi }}{{\mathbf{h}}_{br}}}$ and ${{\mathbf{h}}_{re}^H{\mathbf{\Phi }}{{\mathbf{h}}_{br}}}$ can be recast as ${\mathbf{v}}_p^H{{\mathbf{D}}_{r\varphi}}{{\mathbf{h}}_{br}}$ and ${\mathbf{v}}_p^H{{\mathbf{D}}_{re}}{{\mathbf{h}}_{br}}$, respectively, where ${{\mathbf{D}}_{r\varphi}}$ and ${{\mathbf{D}}_{re}}$ are both diagonal matrices whose diagonal elements are composed from ${\mathbf{h}}_{r\varphi}^H$ and ${\mathbf{h}}_{re}^H$. By utilizing the on-off control scheme, (\ref{SINR usern}) and (\ref{SINR userf}) can be rewritten as
\begin{align}\label{SINR usern hat}
{{\hat \gamma }_n^{ipSIC}} = \frac{{{a_n}P_{BS}^{act}{\kappa ^2}{{\left| {{\mathbf{v}}_p^H{{\mathbf{D}}_{rn}}{{\mathbf{h}}_{br}}} \right|}^2}}}{{{\kappa ^2}\sigma _t^2{{\left\| {{\mathbf{v}}_p^H{{\mathbf{D}}_{rn}}} \right\|}^2} + \varpi P_{BS}^{act}{{\left| {{h_{ipu}}} \right|}^2} + {\sigma ^2}}},
\end{align}
and
\begin{small}
\begin{align}\label{SINR userf hat}
{{\hat \gamma }_f} = \frac{{{a_f}P_{BS}^{act}{\kappa ^2}{{\left| {{\mathbf{v}}_p^H{{\mathbf{D}}_{rf}}{{\mathbf{h}}_{br}}} \right|}^2}}}{{{a_n}P_{BS}^{act}{\kappa ^2}{{\left| {{\mathbf{v}}_p^H{{\mathbf{D}}_{rf}}{{\mathbf{h}}_{br}}} \right|}^2} + {\kappa ^2}\sigma _t^2{{\left\| {{\mathbf{v}}_p^H{{\mathbf{D}}_{rf}}} \right\|}^2} + {\sigma ^2}}},
\end{align}
\end{small}
respectively. Note that the ARIS with on-off control scheme seems similar to the conventional AF relaying, but they are actually  quite different. On one hand, ARIS reflects the incident signals without radio frequency (RF) components \cite{2019BasarRIS} and it is capable of achieving superior performance under the same power budget as shown in Fig. 3 compared to AF relaying. On the other hand, AF relaying equipped with RF components has to simultaneously combines the received signals and precodes the forwarded signals, resulting in a long latency. In this case, two adjacent timeslots will be occupied for users to decode the signals from the relaying link and direct link with the assistance of timing synchronization technology. Unlike the AF relaying, ARIS cannot process the signals and it provides a reflection path with a nanosecond-scale latency \cite{2019BasarRIS}, which means users can receive the signals carrying the same symbol from the direct link and ARIS-aided link in one timeslot \cite{2020RenzoRISvsRelay}.

\begin{lemma} \label{Lemma1}
By utilizing on-off control scheme, the CDF of SINR for user n to decode its own signal with ipSIC can be given by
\begin{small}
\begin{align}\label{CDF SINR n ipsic}
{F_{\hat \gamma _n^{ipSIC}}}\left( x \right) \approx 1 - \frac{2}{{\Gamma \left( Q \right)}}\sum\limits_{d = 1}^D {{G_d}{{\left( {{\Xi _n}x} \right)}^{\frac{Q}{2}}}} {K_Q}\left( {2\sqrt {{\Xi _n}x} } \right),
\end{align}
\end{small}where ${c_n} = {a_n}P_{BS}^{act}{\kappa ^2}$, ${\Xi _n} = \frac{{{v_n} + {\varpi}P_{BS}^{act}{\Omega _{ipu}}{\zeta _d}}}{{{c_n}{\Omega _{br}}{\Omega _{rn}}}}$, ${v_n} = {\kappa ^2}\sigma _t^2Q{\Omega _{rn}} + {\sigma ^2}$, ${G_d} = \frac{{{{\left( {D!} \right)}^2}}}{{{{\zeta _d}}{{\left[ {L_{_D}^\prime \left( {{\zeta _d}} \right)} \right]}^2}}}$ and ${{\zeta _d}}$ represent the weight of Gauss-Laguerre quadrature formula and the d-th zero point of Laguerre polynomial ${L_D}\left( {{\zeta _d}} \right)$ with d = 1,2,3, $\cdots$ ,D, respectively. D presents a complexity accuracy tradeoff parameter and the equal sign in (\ref{CDF SINR n ipsic}) can be established when D approaches infinity.
\end{lemma}
\begin{proof}
Please refer to Appendix A.
\end{proof}
When $\varpi  = 0$, the CDF of SINR for user \emph{n} to decode its own signal with pSIC can be given by
\begin{small}
\begin{align}\label{CDF SINR n psic}
{F_{\hat \gamma _n^{pSIC}}}\left( x \right) = 1 - \frac{2}{{\Gamma \left( Q \right)}}{\left( {\frac{{{v_n}x}}{{{c_n}{\Omega _{br}}{\Omega _{rn}}}}} \right)^{\frac{Q}{2}}}{K_Q}\left( {2\sqrt {\frac{{{v_n}x}}{{{c_n}{\Omega _{br}}{\Omega _{rn}}}}} } \right).
\end{align}
\end{small}
\begin{lemma} \label{Lemma2}
By utilizing on-off control scheme, the CDF of SINR for user f to decode its own signal can be given by
\begin{small}
\begin{align}\label{CDF SINR f}
{F_{{\hat \gamma _f}}}\left( x \right) = 1 - \frac{2}{{\Gamma \left( Q \right)}}{\left( {\frac{{x{\Xi _f}}}{{{c_f} - x{c_n}}}} \right)^{\frac{Q}{2}}}{K_Q}\left( {2\sqrt {\frac{{x{\Xi _f}}}{{{c_f} - x{c_n}}}} } \right),
\end{align}
\end{small}where ${c_f} = {a_f}P_{BS}^{act}{\kappa ^2}$, ${\Xi _f} = \frac{{{v_f}}}{{{\Omega _{br}}{\Omega _{rf}}}}$ and ${v_f} = {\kappa ^2}\sigma _t^2Q{\Omega _{rf}} + {\sigma ^2}$.
\end{lemma}

\section{Secrecy Performance Evaluation}\label{SectionIII}
In this section, the SOP is selected as a pivotal index to investigate the secrecy performance of  ARIS-NOMA networks with the existence of external and internal Eves. Secrecy diversity order is also acquired in the high SNR region to highlight the approximate secrecy features.
\subsection{Statistical Properties for Eavesdropping Channels}
In order to evaluate the impact of Eve on secure communication, statistical properties for eavesdropping channels are derived in the following.
\subsubsection{External Wiretapping Scenario}
According to the principle of on-off control scheme, (\ref{the SINR EE n}) and (\ref{the SINR EE f}) can be transformed into
\begin{align}\label{the SINR EE n new}
{\hat \gamma _{e \to n}^{ipSIC}} = \frac{{{a_n}P_{BS}^{act}{\kappa ^2}{{\left| {{\mathbf{v}}_p^H{{\mathbf{D}}_{re}}{{\mathbf{h}}_{br}}} \right|}^2}}}{{{\kappa ^2}\sigma _t^2{{\left\| {{\mathbf{v}}_p^H{{\mathbf{D}}_{re}}} \right\|}^2} + \varpi P_{BS}^{act}{{\left| {{h_{ipe}}} \right|}^2} + {\sigma _e}^2}},
\end{align}
and
\begin{align}\label{the SINR EE f new}
{\hat \gamma _{e \to f}} = \frac{{{a_f}P_{BS}^{act}{\kappa ^2}{{\left| {{\mathbf{v}}_p^H{{\mathbf{D}}_{re}}{{\mathbf{h}}_{br}}} \right|}^2}}}{{{a_n}P_{BS}^{act}{\kappa ^2}{{\left| {{\mathbf{v}}_p^H{{\mathbf{D}}_{re}}{{\mathbf{h}}_{br}}} \right|}^2} + {\kappa ^2}\sigma _t^2{{\left\| {{\mathbf{v}}_p^H{{\mathbf{D}}_{re}}} \right\|}^2} + {\sigma _e}^2}},
\end{align}respectively.

\begin{lemma} \label{Lemma3}
By utilizing on-off control scheme, the PDF of SINR for the external Eve to wiretap user n's signal with ipSIC can be given by (\ref{PDF SINR EE decode n ipSIC}) which is shown on the top of next page, where ${\Xi _{e1}} = \frac{{{v_{e1}} + {\varpi}P_{BS}^{act}{\Omega _{ipe}}{\zeta _d}}}{{{c_n}{\Omega _{br}}{\Omega _{re}}}}$ and ${v_{e1}} = {\kappa ^2}\sigma _t^2Q{\Omega _{re}} + \sigma _e^2$.
\begin{figure*}[!t]
\begin{small}
\begin{align}\label{PDF SINR EE decode n ipSIC}
{f_{\hat \gamma _{e \to n}^{ipSIC}}}\left( x \right) = \frac{1}{{\Gamma \left( Q \right)}}\left\langle {\sum\nolimits_{d = 1}^D {{G_d}} {{\left( {{\Xi _{e1}}x} \right)}^{\frac{Q}{2}}}\left\{ { - \frac{Q}{x}{K_Q}\left( {2\sqrt {{\Xi _{e1}}x} } \right) + \sqrt {\frac{{{\Xi _{e1}}}}{x}} \left[ {{K_{Q - 1}}\left( {2\sqrt {{\Xi _{e1}}x} } \right) + {K_{Q + 1}}\left( {2\sqrt {{\Xi _{e1}}x} } \right)} \right]} \right\}} \right\rangle .
\end{align}
\end{small}
\begin{small}
\begin{align}\label{PDF SINR EE decode f}
{f_{{\gamma _{e \to f}}}}\left( x \right) \approx \frac{{{c_f}p\left( x \right)}}{{\Gamma \left( Q \right)\left( {{c_f} - x{c_n}} \right)x}}\left\langle {{{\left[ {p\left( x \right)} \right]}^{\frac{{Q - 1}}{2}}}\left\{ {{K_{Q - 1}}\left[ {2\sqrt {p\left( x \right)} } \right] + {K_{Q + 1}}\left[ {2\sqrt {p\left( x \right)} } \right]} \right\} - Q{{\left[ {p\left( x \right)} \right]}^{\frac{Q}{2} - 1}}{K_Q}\left[ {2\sqrt {p\left( x \right)} } \right]} \right\rangle .
\end{align}
\end{small}
\hrulefill \vspace*{0pt}
\end{figure*}
\end{lemma}
\begin{proof}
\emph{We assume that $Z = {\left| {{\mathbf{v}}_p^H{{\mathbf{D}}_{re}}{{\mathbf{h}}_{br}}} \right|^2}$ and the CDF of $\hat \gamma _{e \to n}^{ipSIC}$ can be written as
\begin{align}
{F_{\hat \gamma _{e \to n}^{ipSIC}}}\left( x \right) &= {\rm{Pr}}\left( {\frac{{{c_n}{Z}}}{{{\varpi}P_{BS}^{act}{{\left| {{h_{ipe}}} \right|}^2} + {v_{e1}}}} < x} \right)\notag \\ &= {\rm{Pr}}\left[ {{Z} < \frac{{x\left( {{\varpi}P_{BS}^{act}y + {v_{e1}}} \right)}}{{{c_n}}}} \right]\notag \\ &= \int_0^\infty  {{f_{{{\left| {{h_{ipe}}} \right|}^2}}}\left( y \right){F_{{Z}}}\left[ {\frac{{x\left( {{\varpi}P_{BS}^{act}y + {v_{e1}}} \right)}}{{{c_n}}}} \right]} dy,
\end{align}where ${f_{{{\left| {{h_{ipe}}} \right|}^2}}}\left( y \right) = \frac{1}{{{\Omega _{ipe}}}}{e^{ - y/{\Omega _{ipe}}}}$.
Based on (\ref{a4}), the CDF of cascaded Rayleigh channels can be given as
\begin{align}
{F_Z}\left( z \right) =  1 - \frac{2}{{\Gamma \left( Q \right)}}{\left( {\frac{z}{{{\Omega _{br}}{\Omega _{rn}}}}} \right)^{\frac{Q}{2}}}{K_Q}\left( {2\sqrt {\frac{z}{{{\Omega _{br}}{\Omega _{rn}}}}} } \right).
\end{align}
Referring to (\ref{a5})-(\ref{a6}) in Appendix A, the CDF of $\hat \gamma _{e \to n}^{ipSIC}$ is shown as
\begin{align}
{F_{\hat \gamma _{e \to n}^{ipSIC}}}\left( x \right) = 1 - \frac{2}{{\Gamma \left( Q \right)}}\sum\limits_{d = 1}^D {{G_d}{{\left( {{\Xi _e}x} \right)}^{\frac{Q}{2}}}{K_Q}\left( {2\sqrt {{\Xi _e}x} } \right)} .
\end{align}
By taking the derivative of \emph{x}, we can obtain (\ref{PDF SINR EE decode n ipSIC}) and the proof is completed.}
\end{proof}
When $\varpi $=0, the PDF of SINR for the external Eve to wiretap user \emph{n}'s signal with pSIC is shown as
\begin{normalsize}
\begin{align}\label{PDF SINR EE decode n pSIC}
{f_{\hat \gamma _{e \to n}^{pSIC}}}\left( x \right) =& \frac{1}{{\Gamma \left( Q \right)}}\Xi _{e2}^{\frac{Q}{2}}{x^{\frac{{Q - 1}}{2}}}\left\{ {\Xi _{e2}^{^{\frac{1}{2}}}} \right.\left[ {{K_{Q - 1}}\left( {2\sqrt {{\Xi _{e2}}x} } \right)} \right. \notag \\ &+ \left. {\left. {{K_{Q + 1}}\left( {2\sqrt {{\Xi _{e2}}x} } \right)} \right] - \frac{Q}{{\sqrt x }}{K_Q}\left( {2\sqrt {{\Xi _{e2}}x} } \right)} \right\},
\end{align}
\end{normalsize}where ${\Xi _{e2}} = \frac{{{v_{e1}}}}{{{c_n}{\Omega _{br}}{\Omega _{re}}}}$.
\begin{lemma} \label{Lemma4}
By utilizing on-off control scheme, the PDF of SINR for the external Eve to wiretap user f's signal can be given by (\ref{PDF SINR EE decode f}) as shown on the top of next page, where $p\left( x \right) = \frac{{{\Xi _{e3}}x}}{{{c_f} - x{c_n}}}$ and ${\Xi _{e3}} = \frac{{{v_{e1}}}}{{{\Omega _{br}}{\Omega _{re}}}}$.
\end{lemma}

\subsubsection{Internal Wiretapping Scenario}
In the case of internal wiretapping scenario, the far user \emph{f} is treated as internal Eve since it has the worst condition. According to the principle of on-off control scheme, (\ref{SINR f decode n}) can be transformed into
\begin{align}\label{SINR f decode n hat}
{{\hat \gamma }_{f \to n}} = \frac{{{a_n}P_{BS}^{act}{\kappa ^2}{{\left| {{\mathbf{v}}_p^H{{\mathbf{D}}_{rf}}{{\mathbf{h}}_{br}}} \right|}^2}}}{{{\kappa ^2}\sigma _t^2{{\left\| {{\mathbf{v}}_p^H{{\mathbf{D}}_{rf}}} \right\|}^2} + {\sigma _e}^2}}.
\end{align}

\begin{lemma} \label{Lemma5}
By utilizing on-off control scheme, the PDF of SINR for user f to wiretap user n's signal can be given by
\begin{normalsize}
\begin{align}\label{PDF SINR f decode n}
{f_{{{\hat \gamma }_{f \to n}}}}\left( x \right) =& \frac{1}{{\Gamma \left( Q \right)}}\Xi _{e4}^{\frac{{Q + 1}}{2}}{x^{\frac{{Q - 1}}{2}}}\left\{ {\left[ {{K_{Q - 1}}\left( {2\sqrt {{\Xi _{e4}}x} } \right)} \right.} \right.\notag \\ &\left. { + {K_{Q + 1}}\left( {2\sqrt {{\Xi _{e4}}x} } \right)} \right] - Q{\left( {{\Xi _{e4}}x} \right)^{ - \frac{1}{2}}}\notag \\ &\left. { \times {K_Q}\left( {2\sqrt {{\Xi _{e4}}x} } \right)} \right\},
\end{align}
\end{normalsize}where ${\Xi _{e4}} = \frac{{{v_{e2}}}}{{{c_n}{\Omega _{br}}{\Omega _{rf}}}}$ and ${v_{e2}} = {\kappa ^2}\sigma _t^2Q{\Omega _{rf}} + \sigma _e^2$.
\end{lemma}
\subsection{Secrecy Outage Probability Analysis}
In this subsection, the secrecy outage behaviour is evaluated and the closed-form as well as asymptotic SOP expressions for user \emph{f} and user \emph{n} are both obtained in external/internal wiretapping scenarios.
\subsubsection{External Wiretapping Scenario}
The secrecy rate and target secrecy rate of user $\varphi$ can be defined as ${C_\varphi } = {\left[ {\log \left( {1 + {{\hat \gamma }_\varphi }} \right) - \log \left( {1 + {{\hat \gamma }_{e \to \varphi }}} \right)} \right]^ + }$ and ${R_\varphi }$, respectively, where ${\left[ \Delta  \right]^ + } = \max \left( {\Delta ,0} \right)$ and $\varphi  \in \{ f,n\} $. Hence, the secrecy outage events occur if ${C_\varphi } < {R_\varphi }$ and the SOP expression can be shown as follows
\begin{align}\label{SOP1}
{P_\varphi } &= {\rm{Pr}}\left( {{C_\varphi } < {R_\varphi }} \right) \notag \\ &= {\rm{Pr}}\left[ {\log \left( {\frac{{1 + {{\hat \gamma }_\varphi }}}{{1 + {{\hat \gamma }_{e \to \varphi }}}}} \right) < {R_\varphi }} \right]\notag \\ &= {\rm{Pr}}\left[ {{{\hat \gamma }_\varphi } < {2^{{R_\varphi }}}\left( {1 + {{\hat \gamma }_{e \to \varphi }}} \right) - 1} \right].
\end{align}
Note that the difference between ${\hat \gamma _\phi }$ and ${\hat \gamma _{e \to \varphi }}$ is primarily generated by the variables ${{\mathbf{h}}_{rn}}$ and ${{\mathbf{h}}_{re}}$, which are independent of each other \cite{2020XinweiSecureNOMA,2020YangliangRISPLS}. In addition, the residential interference caused by ipSIC is related to the hardware characteristics of the receiver which can be also regarded as independent variables. Hence, the SOP of users can be further calculated by utilizing marginal distributions in the following steps. According to Lemma \ref{Lemma1} and Lemma \ref{Lemma3}, the SOP for user \emph{n} with ipSIC/pSIC is shown as
\begin{normalsize}
\begin{align}\label{SOP n temp}
P_n^\xi \left( {{R_n}} \right) = \int_0^\infty  {{f_{\hat \gamma _{e \to n}^\xi }}} \left( x \right){F_{\hat \gamma _{n}^\xi }}\left[ {{2^{{R_n}}}\left( {1 + {{\hat \gamma }_{e \to n}}} \right) - 1} \right]dx,
\end{align}
\end{normalsize}where $\xi  \in \{ ipSIC,pSIC\} $. The ${f_{\hat \gamma _{e \to n}^\xi }}\left( x \right)$ and ${F_{\hat \gamma _{n}^\xi }}\left( x \right)$ with ipSIC and pSIC can be acquired from (\ref{PDF SINR EE decode n ipSIC}), (\ref{CDF SINR n ipsic}), (\ref{PDF SINR EE decode n pSIC}) and (\ref{CDF SINR n psic}),respectively.
Through similar analysis, the SOP of user \emph{f} can be shown as follows according to Lemma \ref{Lemma2} and Lemma \ref{Lemma4}.
\begin{normalsize}
\begin{align}\label{SOP f temp}
P_f^\xi \left( {{R_f}} \right) = \int_0^\infty  {{f_{\hat \gamma _{e \to f}^{}}}} \left( x \right){F_{\hat \gamma _{e \to f}^{}}}\left[ {{2^{{R_f}}}\left( {1 + {{\hat \gamma }_{e \to f}}} \right) - 1} \right]dx,
\end{align}
\end{normalsize}where ${f_{\hat \gamma _{e \to f}^{}}}\left( x \right)$ and ${F_{\hat \gamma _{e \to f}^{}}}\left( x \right)$ can be taken from (\ref{PDF SINR EE decode f}) and (\ref{CDF SINR f}), respectively. Obviously, it is difficult to tackle the complex integrations involved in (\ref{SOP n temp}) and (\ref{SOP f temp}). Therefore, approximate expressions of SOP are derived to evaluate the secrecy performance of  ARIS-NOMA networks as follows, whose accuracy is verified by simulation results.
\begin{theorem}\label{Theorem1}
By utilizing on-off control scheme, the closed-form expression of SOP for the external Eve to intercept user n's signal with ipSIC can be given by
\begin{align}\label{SOP EE n ipSIC}
P_n^{ipSIC}\left( {{R_n}} \right) \approx& \sum\limits_{s = 1}^S {{G_s}} \left[ {1 - \frac{2}{{\Gamma \left( Q \right)}}\sum\limits_{d = 1}^D {{G_d}} } \right.\notag \\ &\left. { \times {{\left( {{\Xi _n}{\varepsilon _{n1}}} \right)}^{\frac{Q}{2}}}{K_Q}\left( {2\sqrt {{\Xi _n}{\varepsilon _{n1}}} } \right)} \right],
\end{align}where ${\varepsilon _{n1}} = {2^{{R_n}}}\left[ {1 + \frac{{{a_n}{\rho _e}{\kappa ^2}Q{\Omega _{br}}{\Omega _{re}}}}{{\left( {{\kappa ^2}\sigma _t^2Q{\Omega _{re}}} \right)/\sigma _e^2 + {\varpi}{\rho _e}{\Omega _{ipe}}{\zeta _s} + 1}}} \right] - 1$, ${\rho _e} = \frac{{P_{BS}^{act}}}{{{\sigma _e}^2}}$ represents the received signal-noise radio at Eve, ${G_s} = \frac{{{{\left( {S!} \right)}^2}}}{{{{\zeta _s}}{{\left[ {L_{_S}^\prime \left( {{\zeta _s}} \right)} \right]}^2}}}$ and ${{\zeta _s}}$ represent the weight of Gauss-Laguerre quadrature formula and the d-th zero point of Laguerre polynomial ${L_S}\left( {{\zeta _s}} \right)$ with s = 1,2,3, $\cdots$ ,S, respectively. S presents a complexity accuracy tradeoff parameter and the equal sign in (\ref{SOP EE n ipSIC}) can be established when S approaches infinity.
\end{theorem}
\begin{proof}
Please refer to Appendix B.
\end{proof}When $\varpi $=0, the closed-form expression of SOP for the external Eve to wiretap user \emph{n}'s signal with pSIC can be given by
\begin{normalsize}
\begin{align}\label{SOP EE n pSIC}
P_n^{pSIC} = 1 - \frac{2}{{\Gamma \left( Q \right)}}{\left( {\frac{{{\varepsilon _{n2}}{v_n}}}{{{c_n}{\Omega _{br}}{\Omega _{rn}}}}} \right)^{\frac{Q}{2}}}{K_Q}\left( {2\sqrt {\frac{{{\varepsilon _{n2}}{v_n}}}{{{c_n}{\Omega _{br}}{\Omega _{rn}}}}} } \right),
\end{align}
\end{normalsize}where ${\varepsilon _{n2}} = {2^{{R_n}}}\left[ {1 + \frac{{{a_n}{\rho _e}{\kappa ^2}Q{\Omega _{br}}{\Omega _{re}}}}{{\left( {{\kappa ^2}\sigma _t^2Q{\Omega _{re}}} \right)/\sigma _e^2 + 1}}} \right] - 1$.
\begin{theorem}\label{Theorem2}
By utilizing on-off control scheme, the closed-form expression of SOP for the external Eve to intercept user f's signal can be given by
\begin{normalsize}
\begin{align}\label{SOP EE f}
{P_f} = 1 - \frac{2}{{\Gamma \left( Q \right)}}{\left( {\frac{{{\varepsilon _f}{\Xi _f}}}{{{c_f} - {\varepsilon _f}{c_n}}}} \right)^{\frac{Q}{2}}}{K_Q}\left( {2\sqrt {\frac{{{\varepsilon _f}{\Xi _f}}}{{{c_f} - {\varepsilon _f}{c_n}}}} } \right),
\end{align}
\end{normalsize}where ${\varepsilon _f} = {2^{{R_f}}}\left( {1 + \frac{{{a_f}{\rho _e}{\kappa ^2}Q{\Omega _{br}}{\Omega _{re}}}}{{\left( {{\kappa ^2}\sigma _t^2Q{\Omega _{re}}} \right)/{\sigma _e}^2 + {a_n}{\rho _e}{\kappa ^2}Q{\Omega _{br}}{\Omega _{re}} + 1}}} \right) - 1$.
\end{theorem}

\subsubsection{Internal Wiretapping Scenario}
In the case of internal wiretapping scenarios, the far user \emph{f} with poor channel condition is considered as an internal Eve intending to wiretap the legitimate signals of user \emph{n}. The secrecy rate for user \emph{f} to intercept user \emph{n}'s information is given by ${C_{f \to n}} = {\left[ {\log \left( {1 + \hat \gamma _n^{}} \right) - \log \left( {1 + {{\hat \gamma }_{f \to n}}} \right)} \right]^ + }$. Hence, the SOP expression for user \emph{n} is shown as follows
\begin{normalsize}
\begin{align}\label{SOP f n temp}
P_{f \to n}^\xi \left( {{R_n}} \right) &= {\rm{Pr}}\left[ {\log \left( {1 + \hat \gamma _n^\xi } \right) - \log \left( {1 + {{\hat \gamma }_{f \to n}}} \right) < {R_n}} \right]\notag \\ &= {\rm{Pr}}\left[ {\hat \gamma _n^\xi  < {2^{{R_n}}}\left( {1 + {{\hat \gamma }_{f \to n}}} \right) - 1} \right]\notag \\ &= \int_0^\infty  {{f_{{{\hat \gamma }_{f \to n}}}}\left( x \right)} {F_{\hat \gamma _n^\xi }}\left[ {{2^{{R_n}}}\left( {1 + {{\hat \gamma }_{f \to n}}} \right) - 1} \right]dx,
\end{align}
\end{normalsize}where ${{f_{{{\hat \gamma }_{f \to n}}}}\left( x \right)}$ is taken from (\ref{PDF SINR f decode n}) and ${F_{\hat \gamma _n^\xi }}\left( x \right)$ with ipSIC as well as pSIC can be acquired from (\ref{CDF SINR n ipsic}) and (\ref{CDF SINR n psic}), respectively. Integrations in (\ref{SOP f n temp}) is still intractable to solve, so we use an approximation method to deal with the problem.
\begin{theorem}\label{Theorem3}
By utilizing on-off control scheme, the closed-form expression of SOP for user f to intercept user n's signal can be given by
\begin{normalsize}
\begin{align}\label{SOP f n ipSIC}
P_{f \to n}^{ipSIC}\left( {{R_n}} \right) =& \sum\limits_{d = 0}^D {{G_d}} \left[ {1 - \frac{2}{{\Gamma \left( Q \right)}}{{\left( {{\varepsilon _{f \to n}}{\Xi _{e5}}} \right)}^{\frac{Q}{2}}}} \right.\notag \\ &\times \left. {{K_Q}\left( {2\sqrt {{\varepsilon _{f \to n}}{\Xi _{e5}}} } \right)} \right],
\end{align}
\end{normalsize}where ${\varepsilon _{f \to n}} = {2^{{R_n}}}\left[ {1 + \frac{{{a_n}{\rho _e}{\kappa ^2}Q{\Omega _{br}}{\Omega _{rf}}}}{{\left( {{\kappa ^2}\sigma _t^2Q{\Omega _{rf}}} \right)/{\sigma _e}^2 + 1}}} \right] - 1$ and ${\Xi _{e5}} = \frac{{{v_n} + {\varpi}P_{BS}^{act}{\Omega _{ipe}}{\zeta _d}}}{{{c_n}{\Omega _{br}}{\Omega _{rn}}}}$.
\end{theorem}When $\varpi $=0, the closed-form expression of SOP for user \emph{f} to wiretap user \emph{n}'s signal can be shown as
\begin{align}\label{SOP f n pSIC}
P_{f \to n}^{pSIC}\left( {{R_n}} \right) =& 1 - \frac{2}{{\Gamma \left( Q \right)}}{\left( {\frac{{{\varepsilon _{f \to n}}{v_n}}}{{{c_n}{\Omega _{br}}{\Omega _{rn}}}}} \right)^{\frac{Q}{2}}}\notag \\ &\times{K_Q}\left( {2\sqrt {\frac{{{\varepsilon _{f \to n}}{v_n}}}{{{c_n}{\Omega _{br}}{\Omega _{rn}}}}} } \right).
\end{align}
\subsection{Secrecy Diversity Order}
To further confirm the secrecy performance of ARIS-NOMA networks, the asymptotic expressions of SOP are attained in high SNR region. Define the secrecy diversity order as follows
\begin{align}\label{div define}
div =  - \mathop {\lim }\limits_{\rho  \to \infty } \frac{{\log \left[ {{P_{asy}}\left( \rho  \right)} \right]}}{{\log \rho }},
\end{align}where ${{P_{asy}}\left( \rho  \right)}$ refers to the asymptotic SOP with parameter $\rho$. The asymptotic SOPs of users under external as well as internal eavesdropping scenarios are both researched in the following.
\begin{corollary}\label{Corollary1}
By utilizing on-off control, the asymptotic SOP for external Eve to decode user n's signals with ipSIC when $\rho  \to \infty $ can be given by
\begin{normalsize}
\begin{align}\label{ASY SOP EE n ipSIC}
P_{n,asy}^{ipSIC}\left( {{R_n}} \right) =& \sum\limits_{s = 1}^S {{G_s}} \left[ {1 - \frac{2}{{\Gamma \left( Q \right)}}\sum\limits_{d = 1}^D {{G_d}} {{\left( {\frac{{{\varepsilon _{n1}}{\Omega _{ipu}}{\zeta _d}}}{{{a_n}{\kappa ^2}{\Omega _{br}}{\Omega _{rn}}}}} \right)}^{\frac{Q}{2}}}} \right.\notag \\ &\left. { \times {K_Q}\left( {2\sqrt {\frac{{{\varepsilon _{n1}}{\Omega _{ipu}}{\zeta _d}}}{{{a_n}{\kappa ^2}{\Omega _{br}}{\Omega _{rn}}}}} } \right)} \right],
\end{align}
\end{normalsize}
\end{corollary}
\begin{proof}
\emph{When $\rho  \to \infty $, the SINR $\hat \gamma _n^{ipSIC} = \frac{{{a_n}P_{BS}^{act}{\kappa ^2}{{\left| {{\mathbf{v}}_p^H{{\mathbf{D}}_{rn}}{{\mathbf{h}}_{br}}} \right|}^2}}}{{{\kappa ^2}\sigma _t^2{{\left\| {{\mathbf{v}}_p^H{{\mathbf{D}}_{rn}}} \right\|}^2} + \varpi P_{BS}^{act}{{\left| {{h_{ipu}}} \right|}^2} + {\sigma ^2}}}$ can be approximatively expressed as $\hat \gamma _{n,asy}^{ipSIC} = \frac{{{a_n}{\kappa ^2}{{\left| {{\mathbf{v}}_p^H{{\mathbf{D}}_{rn}}{{\mathbf{h}}_{br}}} \right|}^2}}}{{{{\left| {\varpi}{{h_{ipu}}} \right|}^2}}}$. The CDF of cascaded Rayleigh channels can be given by
\begin{align}
{F_Z}\left( z \right) =  1 - \frac{2}{{\Gamma \left( Q \right)}}{\left( {\frac{z}{{{\Omega _{br}}{\Omega _{rn}}}}} \right)^{\frac{Q}{2}}}{K_Q}\left( {2\sqrt {\frac{z}{{{\Omega _{br}}{\Omega _{rn}}}}} } \right).
\end{align}
Hence, the CDF of $Z = {\left| {{\mathbf{v}}_p^H{{\mathbf{D}}_{rn}}{{\mathbf{h}}_{br}}} \right|^2}$ can be rewritten as
\begin{normalsize}
\begin{align}\label{ASY CDF Z}
F_Z^{asy}\left( x \right) =& \Pr \left( {\frac{{{a_n}{\kappa ^2}Z}}{{\varpi {{\left| {{h_{ipu}}} \right|}^2}}} < x} \right)\notag \\ =& 1 - \frac{2}{{\Gamma \left( Q \right)}}{\left( {\frac{{\varpi xy}}{{{a_n}{\kappa ^2}{\Omega _{br}}{\Omega _{rn}}}}} \right)^{\frac{Q}{2}}}\notag \\ &\times{K_Q}\left( {2\sqrt {\frac{{\varpi xy}}{{{a_n}{\kappa ^2}{\Omega _{br}}{\Omega _{rn}}}}} } \right).
\end{align}
\end{normalsize}
According to (\ref{a1}), the CDF of $\hat \gamma _{n,asy}^{ipSIC}$ is shown as
\begin{normalsize}
\begin{align}\label{ASY CDF Z temp}
{F_{\hat \gamma _{n,asy}^{ipSIC}}}\left( x \right) =& \int_0^\infty  {{f_{{{\left| {{h_{ipu}}} \right|}^2}}}\left( y \right)} F_Z^{asy}\left( x \right)dy \notag \\
=& 1 - \frac{2}{{{\Omega _{ipu}}\Gamma \left( Q \right)}}\int_0^\infty  {{e^{ - \frac{y}{{{\Omega _{ipu}}}}}}} {\left( {\frac{{\varpi xy}}{{{a_n}{\kappa ^2}{\Omega _{br}}{\Omega _{rn}}}}} \right)^{\frac{Q}{2}}}\notag \\ &\times {K_Q}\left( {2\sqrt {\frac{{\varpi xy}}{{{a_n}{\kappa ^2}{\Omega _{br}}{\Omega _{rn}}}}} } \right)dy.
\end{align}
\end{normalsize}By replacing $y/{\Omega _{ipu}}$ with t and using Gauss-Laguerre quadrature formula, we can obtain
\begin{normalsize}
\begin{align}\label{ASY CDF Z}
{F_{\hat \gamma _{n,asy}^{ipSIC}}}\left( x \right) =& 1 - \frac{2}{{\Gamma \left( Q \right)}}\sum\limits_{d = 1}^D {{G_d}} {\left( {\frac{{\varpi x{\Omega _{ipu}}{\zeta _d}}}{{{a_n}{\kappa ^2}{\Omega _{br}}{\Omega _{rn}}}}} \right)^{\frac{Q}{2}}}\notag \\ &\times{K_Q}\left( {2\sqrt {\frac{{\varpi x{\Omega _{ipu}}{\zeta _d}}}{{{a_n}{\kappa ^2}{\Omega _{br}}{\Omega _{rn}}}}} } \right).
\end{align}
\end{normalsize}With the assistance of (\ref{b4}) and some straightforward calculations, we can acquire (\ref{ASY SOP EE n ipSIC}) and the proof is completed.}
\end{proof}
\begin{remark}\label{remark1}
Upon substituting (\ref{ASY SOP EE n ipSIC}) into (\ref{div define}), the secrecy diversity order for external Eve to intercept user n's information with ipSIC equals zero, which is caused by the residual interference from ipSIC.
\end{remark}
When $\varpi $=0, the asymptotic SOP for external Eve to decode user \emph{n}'s signals with pSIC can be given in two cases, i.e., \emph{Q} = 1 and \emph{Q} $ > $ 1 which are respectively shown as
\begin{normalsize}
\begin{align}\label{ASY SOP EE n pSIC Q1}
P_{n,asy}^{pSIC}\left( {{R_n}} \right) = { -{\left( {\frac{{{v_n}{\varepsilon _{n2}}}}{{{c_n}{\Omega _{br}}{\Omega _{rn}}}}} \right)} \ln \left( {\frac{{{v_n}{\varepsilon _{n2}}}}{{{c_n}{\Omega _{br}}{\Omega _{rn}}}}} \right)},
\end{align}
\end{normalsize}
and
\begin{normalsize}
\begin{align}\label{ASY SOP EE n pSIC Q2}
P_{n,asy}^{pSIC}\left( {{R_n}} \right) = \frac{{{v_n}{\varepsilon _{n2}}}}{{{c_n}{\Omega _{br}}{\Omega _{rn}}\left( {Q - 1} \right)}}.
\end{align}
\end{normalsize}
\begin{proof}
\emph{According to \emph{\cite{ZhiguoSimpleDesign}}, the modified Bessel function of the second kind can be approximately expressed as ${K_Q}\left( x \right) \approx \frac{1}{x} + \frac{x}{2}\ln \left( {\frac{x}{2}} \right)$ and ${K_Q}\left( x \right) \approx \frac{1}{2}\left[ {\frac{{\left( {Q - 1} \right)!}}{{{{\left( {x/2} \right)}^Q}}} - \frac{{\left( {Q - 2} \right)!}}{{{{\left( {x/2} \right)}^{Q - 2}}}}} \right]$, respectively. As a consequence, the CDF expressions for user n in high SNR region can be acquired by substituting the asymptotic Bessel functions into (\ref{CDF SINR n psic}), which are shown as
\begin{align}\label{ASY CDF EE n pSIC Q1 proof}
F_{\hat \gamma _n^{pSIC}}^{asy}\left( x \right) = {-{\left( {\frac{{{v_n}x}}{{{c_n}{\Omega _{br}}{\Omega _{rn}}}}} \right)}\ln \left( {\frac{{{v_n}x}}{{{c_n}{\Omega _{br}}{\Omega _{rn}}}}} \right)},
\end{align}
and
\begin{align}\label{ASY CDF EE n pSIC Q2 proof}
F_{\hat \gamma _n^{pSIC}}^{asy}\left( x \right) = \frac{{{v_n}x}}{{{c_n}{\Omega _{br}}{\Omega _{rn}}\left( {Q - 1} \right)}},
\end{align}respectively.
With the assistance of (\ref{b4}) and simple calculations, we can attain (\ref{ASY SOP EE n pSIC Q1}) and (\ref{ASY SOP EE n pSIC Q2}). The proof is completed.}
\end{proof}
\begin{corollary}\label{Corollary2}
By utilizing on-off control, the asymptotic SOP for external Eve to decode user f's signals when $\rho  \to \infty $ can be given by
\begin{normalsize}
\begin{align}\label{ASY SOP EE f Q1}
{P_{f,asy}}\left( {{R_f}} \right) = {-{\left( {\frac{{{\Xi _f}{\varepsilon _f}}}{{{c_f} - {c_n}{\varepsilon _f}}}} \right)}\ln \left( {\frac{{{\Xi _f}{\varepsilon _f}}}{{{c_f} - {c_n}{\varepsilon _f}}}} \right)},
\end{align}
\end{normalsize}
and
\begin{normalsize}
\begin{align}\label{ASY SOP EE f Q2}
{P_{f,asy}}\left( {{R_f}} \right) = \frac{{{\Xi _f}{\varepsilon _f}}}{{\left( {{c_f} - {c_n}{\varepsilon _f}} \right)\left( {Q - 1} \right)}},
\end{align}
\end{normalsize}
for Q = 1 and Q $>$ 1, respectively.
\end{corollary}

\begin{remark}\label{remark2}
Upon substituting (\ref{ASY SOP EE n pSIC Q1}) and (\ref{ASY SOP EE n pSIC Q2}) into (\ref{div define}), the secrecy diversity order for external Eve to intercept user n's information with pSIC equals 1. Similarly, the secrecy diversity order for external Eve to intercept user f's information also equals 1 by substituting (\ref{ASY SOP EE f Q1}) and (\ref{ASY SOP EE f Q2}) into (\ref{div define}) for arbitrary value of Q. This remark indicates that although the thermal noise can impair the secrecy outage behavior of ARIS-NOMA networks, it cannot shape the trend of the curves like residual interference.
\end{remark}
\subsection{Secrecy Throughput Analysis}
Based on the derived SOP expressions above, the secrecy throughput for user $\varphi $ in external eavesdropping scenarios with ipSIC and pSIC is defined as follows
\begin{align}\label{SST-EE}
T_\varphi ^\xi \left( {{R_\varphi }} \right) = \left[ {1 - P_\varphi ^\xi \left( {{R_\varphi }} \right)} \right]{R_\varphi },
\end{align}where $\xi  \in \left\{ {ipSIC,pSIC} \right\}$, ${P_\varphi ^\xi \left( {{R_\varphi }} \right)}$ can be obtained from (\ref{SOP EE n ipSIC}), (\ref{SOP EE n pSIC}) and (\ref{SOP EE f}) in Section \ref{SectionIII}-B, respectively.

Similarly, the secrecy throughput for user \emph{n} in internal eavesdropping scenarios with ipSIC and pSIC is given by
\begin{align}\label{SST-IE}
T_{f \to n}^\xi \left( {{R_n}} \right) = \left[ {1 - P_{f \to n}^\xi \left( {{R_n}} \right)} \right]{R_n},
\end{align}where $\xi  \in \left\{ {ipSIC,pSIC} \right\}$, ${P_{f \to n}^\xi \left( {{R_n}} \right)}$ can be obtained from (\ref{SOP f n ipSIC}) and (\ref{SOP f n pSIC}) in Section \ref{SectionIII}-B, respectively.

\section{Numerical Results}\label{SectionIV}
\begin{table}[h]
\caption{The table of Monte Carlo simulation parameters}
\begin{tabular}{ll} 
\toprule
Parameter types &Values \\
\midrule
The pathloss factors                       & $\alpha  = 2$, $\beta  =  - 30$ dB                                                       \\
The AWGN power                             & ${\sigma ^2} = \sigma _e^2 =  - 55$ dBm                                                        \\
The hardware power consumption &  ${P_{PS}} = {P_{DC}} = - 10$ dBm  \\

The power allocation coefficients & \begin{tabular}[c]{@{}l@{}}${a_f} = 0.7$, ${a_n} = 0.3$\end{tabular}  \\
The target secrecy rates for users         & \begin{tabular}[c]{@{}l@{}}$R_f = R_n = 0.05$ BPCU\end{tabular}        \\

The distance from BS to ARIS               & ${d_{br}} = 20{\text{ m}}$                                                            \\
The distance from ARIS to users/Eves       & \begin{tabular}[c]{@{}l@{}}${d_{rf}} = {d_{re}} = 20{\text{ m}}$\\
${d_{rn}} = 10{\text{ m}}$ \end{tabular} \\
  \bottomrule
  \end{tabular}
  \label{tbl:table1}
\end{table}

In this section, numerical results are presented to confirm the accuracy of the closed-form expressions under both external and internal eavesdropping scenarios. To facilitate
notational presentation, the simulation parameters utilized are given in Table I \cite{2021CaihongSecureNOMA,CunhuaPanARIS}, where BPCU indicates the abbreviation for bit per
channel use. Note that the values of ${\left| {{h_{ipu}}} \right|^2}$ and ${\left| {{h_{ipe}}} \right|^2}$ closely relate to the power of the signal received by the users. Specifically, about $1/7$, $1/70$ and $1/700$ of user \emph{m}'s signal power are taken as the residential interference during user \emph{n}'s SIC process \cite{2022XinweiYueRISNOMA}. In addition, since the reconfigurable elements on ARIS are small and closely spaced [18], the fading characteristics of the cascaded channels corresponding to each reconfigurable element differ very little from each other, so all the attenuation factors are set to the same value. In future studies, especially for large-scale ARIS scenarios, we will relax this assumption and consider the impact of the variation of the attenuation factor on each ARIS element on the secrecy performance of ARIS-NOMA networks. To highlight the secrecy performance of  ARIS-aided NOMA networks, the PRIS-aided networks, AF relay, and half-duplex (HD)/FD-DF relay communication schemes are considered as benchmarks, which share the same total power budget as ARIS-aided NOMA networks.

\subsection{External Wiretapping Scenario}

\begin{figure}[t!]
    \begin{center}
        \includegraphics[width=2.93in,  height=2.2in]{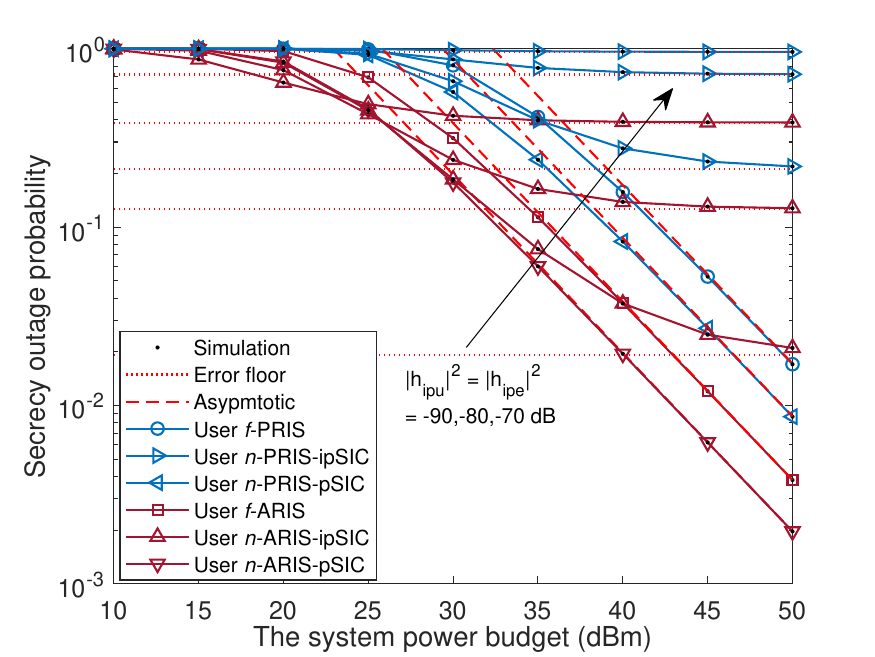}
        \caption{The SOP versus system power budget under ARIS-NOMA and PRIS-NOMA networks in external eavesdropping scenarios, with \emph{M} = 40, \emph{P} = 2 and \emph{Q} = 20.}
        \label{EE_SOP_1}
    \end{center}
\end{figure}

Fig. \ref{EE_SOP_1} plots the SOP versus system power budget under external eavesdropping scenario, where the secrecy performance of PRIS-NOMA networks is taken into account for comparison. The theoretical analysis curves for user \emph{f} and user \emph{n} with ipSIC/pSIC are obtained from (\ref{SOP EE f}), (\ref{SOP EE n ipSIC}) and (\ref{SOP EE n pSIC}), respectively. It is obvious to see that the Monte Carlo simulation results are consistent with the analytical results during the entire range of power budget. In the high SNR region, analysis curves can converge exactly to the asymptotic SOP lines based on (\ref{ASY SOP EE f Q2}), (\ref{ASY SOP EE n ipSIC}) and (\ref{ASY SOP EE n pSIC Q2}), which also manifests the validity of the theoretical expressions. We can observe that ARIS-NOMA networks can achieve lower SOP than PRIS-NOMA networks. This is because that active reconfigurable elements can further amplify the input radio signals $\kappa $ times larger by allocating a portion of the system power budget to these components. Hence, the higher SNR can be received at users and thus strengthen the secure communication. Moreover, the error floors appear for  ARIS-NOMA networks with ipSIC. This is due to the residential interference caused by ipSIC, which also matches with the conclusions in \textbf{Remark \ref{remark1}}.

\begin{figure}[t!]
    \begin{center}
        \includegraphics[width=2.93in,  height=2.2in]{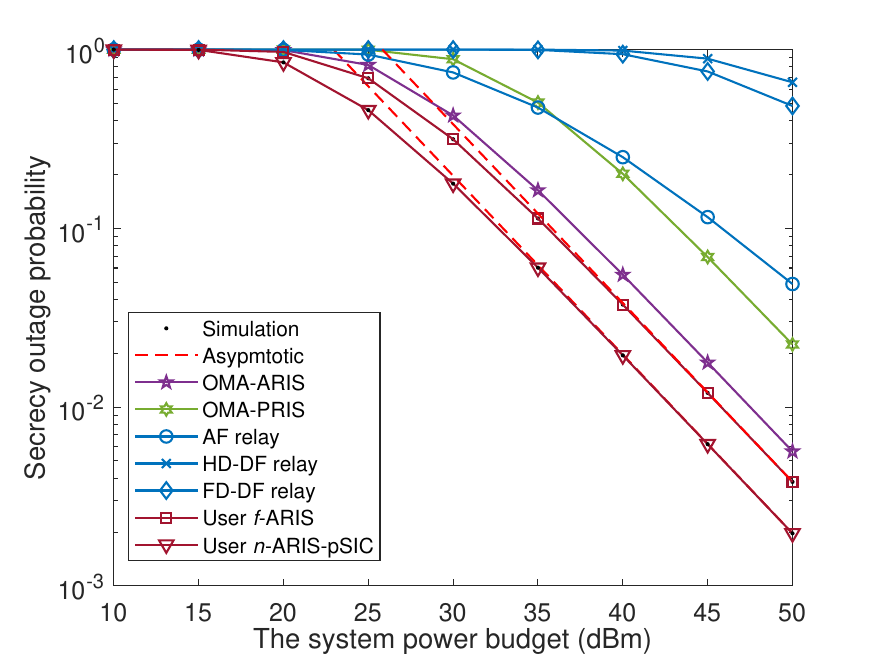}
        \caption{The SOP versus system power budget under ARIS-NOMA, ARIS/PRIS-OMA and conventional relaying schemes in external eavesdropping scenarios, with \emph{M} = 40, \emph{P} = 2, \emph{Q} = 20 and $R_{OMA}$ = 0.1 BPCU.}
        \label{EE_SOP_2}
    \end{center}
\end{figure}

Fig. \ref{EE_SOP_2} plots the SOP versus system power budget under external eavesdropping scenario, where the secrecy performance of ARIS/PRIS-OMA networks, AF relay and HD/FD-DF relay schemes are surveyed as benchmarks. One can observe from the figure that the secrecy performance of  ARIS-NOMA networks exceeds that of ARIS-OMA, AF relay, HD-DF and FD-DF communication schemes. The reasons are explained as follows: 1) The higher utilization rate of spectrum can be acquired in  ARIS-NOMA networks since NOMA is able to superpose multiple signals together within the same resource block; and 2) ARIS is able to adjust the phase shifts and amplitudes of the incident signals, which remarkably improve the electromagnetic environment and enhance the quality of channel conditions around.

\begin{figure}[t!]
    \begin{center}
        \includegraphics[width=2.93in,  height=2.2in]{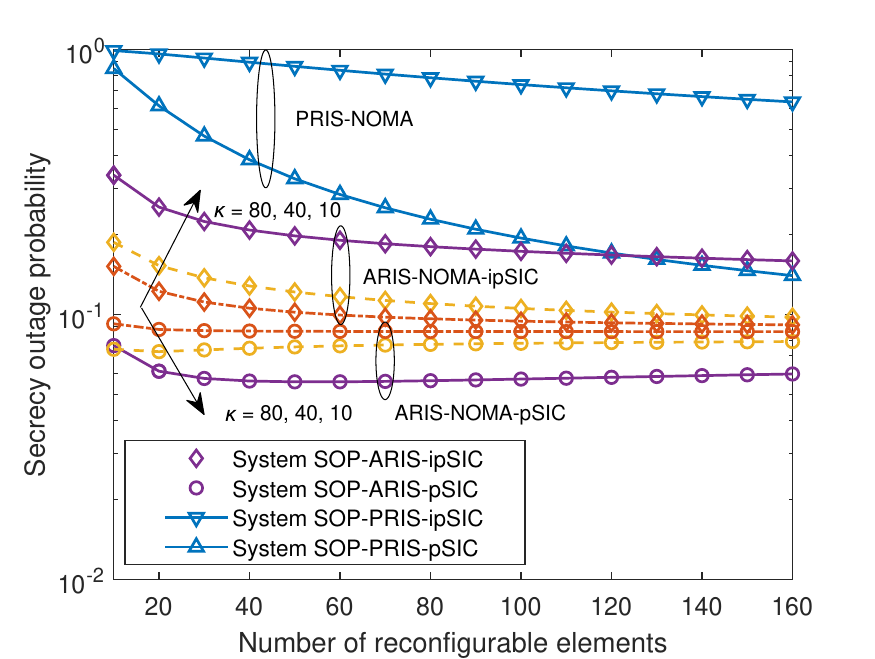}
        \caption{The system SOP versus different number of reconfigurable elements under external eavesdropping scenario, with ${P_{tot}} =$ -40 dBm, $\sigma _t^2 =$ -40 dBm and ${\left| {{h_{ipu}}} \right|^2} = {\left| {{h_{ipe}}} \right|^2} =$ -80 dB.}
        \label{SOP_diff_M}
    \end{center}
\end{figure}

Fig. \ref{SOP_diff_M} plots the system SOP versus different number of reconfigurable elements under external eavesdropping scenario. The analysis curves are plotted according to (\ref{SOP EE f}), (\ref{SOP EE n ipSIC}), (\ref{SOP EE n pSIC}), respectively. We can observe that the system SOP of  ARIS-NOMA networks drops first and then level off at stable values with the growth of the number of active reconfigurable elements. This is due to the fact that introducing more active components are able to improve the freedom of spatial design, however, excessive reconfigurable elements can also mix lots of thermal noise in the received signals, which is not conducive to users decoding their own signals and counteracts the channel gain brought by the spatial freedom. Another observation is that larger values of $\kappa $ lead to the decline in the system SOP for  ARIS-NOMA networks with ipSIC. The main reason behind this phenomenon is that more power will be occupied at  ARIS since the value of $\kappa $ rises gradually and it can directly weaken the unfavorable effect caused by residential interference. On the contrary, decreasing $\kappa $ is beneficial to the system secrecy performance of  ARIS-NOMA with pSIC. This is because the thermal noise brought by active elements holds the main station of the received interference in pSIC situation. Obviously, applying smaller $\kappa $ can effectively lower the detrimental noise and provide users with superior SNR to resist eavesdropping.

\begin{figure}[t!]
    \begin{center}
        \includegraphics[width=2.93in,  height=2.2in]{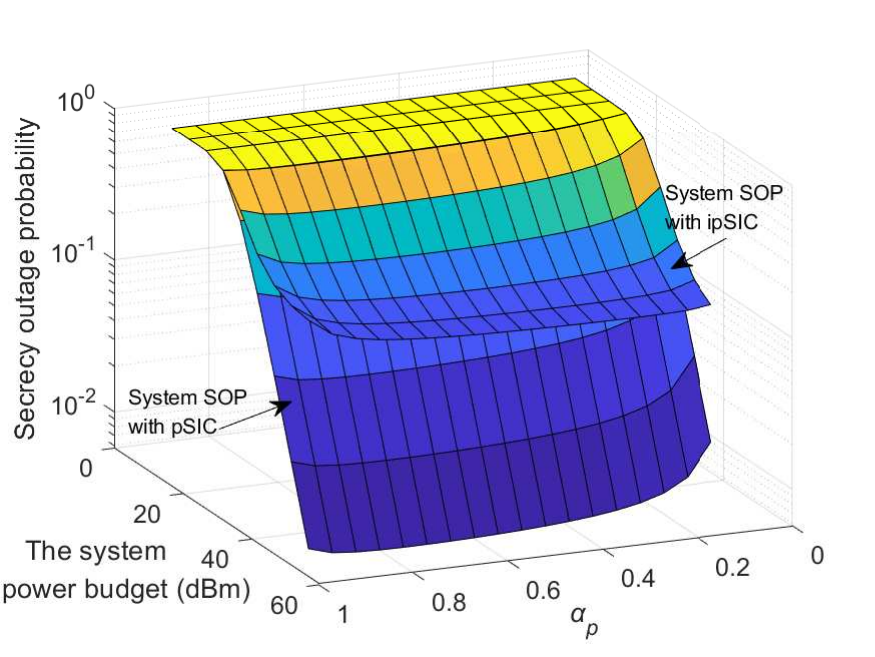}
        \caption{The SOP versus system power budget under external eavesdropping scenario, with \emph{M} = 40, \emph{P} = 2, \emph{Q} = 20 and ${\left| {{h_{ipu}}} \right|^2} = {\left| {{h_{ipe}}} \right|^2} = $ -80 dB.}
        \label{SOP_diff_powerallocation_EE}
    \end{center}
\end{figure}
Fig. \ref{SOP_diff_powerallocation_EE} plots the SOP versus system power budget under external eavesdropping scenario. We define ${\alpha _P} \in \left( {0,1} \right)$ as the power offset factor and the power allocation coefficients for user \emph{f} and user \emph{n} can be given by ${a_f} = {\alpha _P}$ and ${a_n} = 1 - {\alpha _P}$, respectively. As can be observed from the figure that the system secrecy performance of  ARIS-NOMA networks dips gradually with the increase of ${\alpha _P}$. The reason is that fairness can be achieved for both non-orthogonal users by allocating the far user \emph{f} with more power in ARIS-NOMA networks, which is helpful to preserve the secure transmission. Another observation is that the system SOP begin to climbing since ${\alpha _P}$ reaches 0.8 and even more. This is because that overlarge value of ${\alpha _P}$ will seriously encroach on the power allocated to the near user \emph{n} and break the fairness between both users. As a consequence, the system SOP start to show signs of getting worse. This suggests that it is pivotal to set suitable power allocation coefficients in  ARIS-NOMA networks to guarantee secure communication.

\begin{figure}[t!]
    \begin{center}
        \includegraphics[width=2.93in,  height=2.2in]{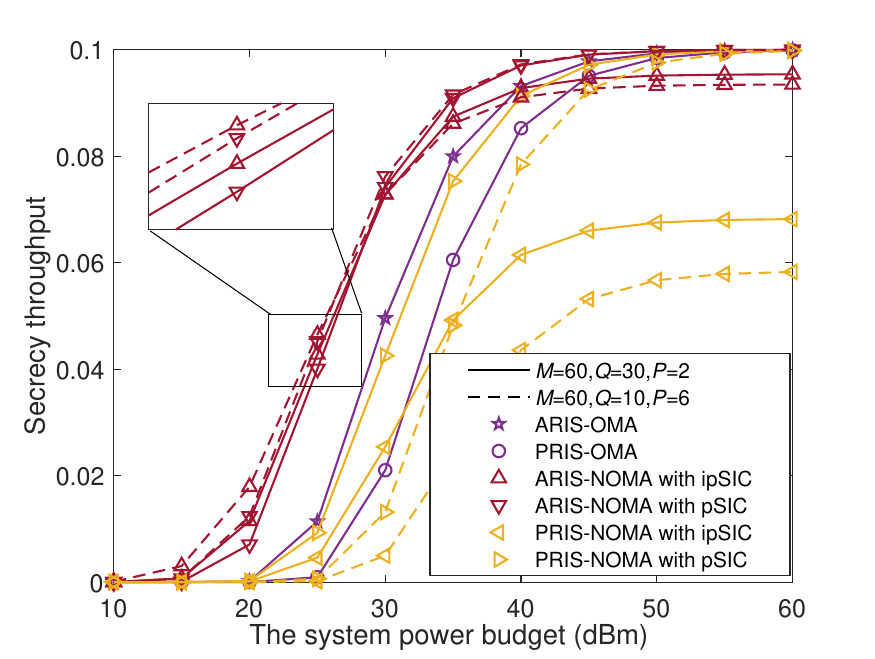}
        \caption{The system secrecy throughput versus system power budget under external eavesdropping scenario, with \emph{M} = 60, $\kappa = 20$ and ${\left| {{h_{ipu}}} \right|^2} = {\left| {{h_{ipe}}} \right|^2} = $ -80 dB.}
        \label{SST_EE}
    \end{center}
\end{figure}
Fig. \ref{SST_EE} plots the system secrecy throughput versus system power budget under external eavesdropping scenario. The curves for  ARIS-NOMA with ipSIC/pSIC can be obtained from (\ref{SOP EE n ipSIC}), (\ref{SOP EE n pSIC}), (\ref{SOP EE f}) and (\ref{SST-EE}), respectively. It is observed from the figure that ARIS-NOMA networks are capable of achieving higher system secrecy throughput than PRIS-NOMA networks before convergence. This is because the amplification function offers ARIS-NOMA superior secrecy outage performance, which leads to the larger secrecy throughput. Another observation is that the throughput of ARIS-NOMA decreases when the parameters vary from \emph{Q} = 10, \emph{P} = 6 to \emph{Q} = 30, \emph{P} = 2 at the beginning and increase in the high system power budget region. This phenomenon can be explained that setting more active components to $one$ will deepen the interference caused by thermal noise since the distortion from ipSIC is not apparent at first. However, the negative effects brought by residual interference gradually dominate with the sharp rise of the system power budget, and adjusting more elements to $one$ benefits alleviating the harmful influence of ipSIC. That is, values of \emph{Q} and \emph{P} are set largely depending on the prevailing type of interference. This kind of mutual suppression relationship between thermal noise and residual interference is also manifested in the analysis of Fig. \ref{SOP_diff_M}, which makes the gap of ARIS-NOMA less obvious compared to that of PRIS-NOMA with the variation of \emph{Q} and \emph{P}.

\begin{figure}[t!]
    \begin{center}
        \includegraphics[width=2.93in,  height=2.2in]{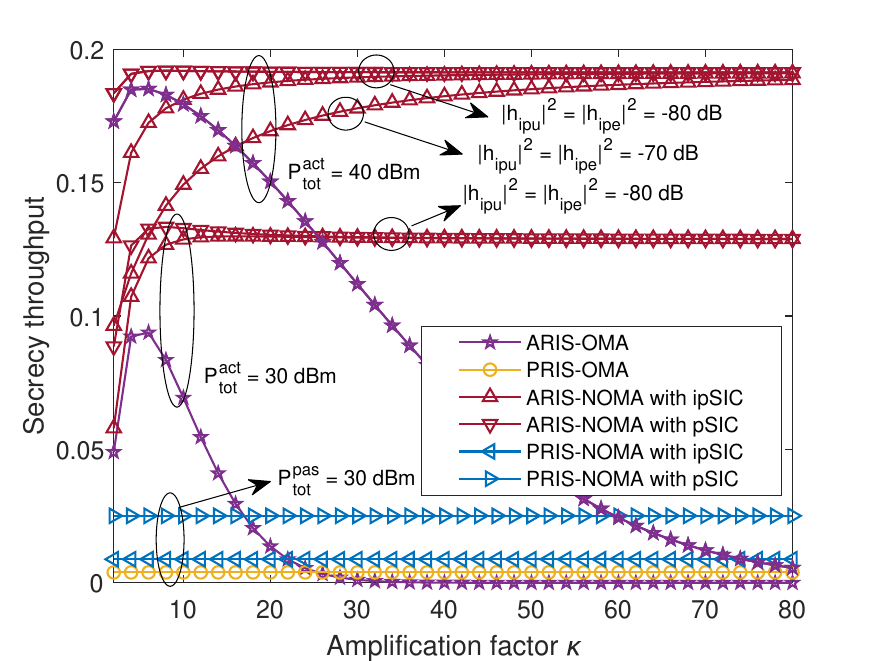}
        \caption{The system secrecy throughput versus amplification factor under external eavesdropping scenario, with \emph{M} = 40 and $R_{n}$ = $R_{f}$ = 0.1 BPCU.}
        \label{SST_EE_diff_kappa}
    \end{center}
\end{figure}
Fig. \ref{SST_EE_diff_kappa} plots the system secrecy throughput versus amplification factor under external eavesdropping scenario. One can observe from the figure that the secrecy throughput is able to realize a remarkable development via expanding the scope of system power budget. This is due to the fact that sufficient power budget is helpful to improve the received SNR at users and thus enhance the anti-eavesdropping capability, which is also in accordance with the analysis of Fig. \ref{EE_SOP_1} and Fig. \ref{EE_SOP_2}. Another observation is that the gap in terms of secrecy throughput between ipSIC and pSIC is negligible with the growth of amplification factor $\kappa$. The reason is that the effect of thermal noise on throughput performance will exceed that of residual interference and the latter can even be neglected when $\kappa$ keeps rising. This can also be utilized to explain the phenomenon that a larger $\kappa$ is required to bridge the gap caused by ipSIC when the residual interference changes from -80 dBm to the severe -70 dBm.

\subsection{Internal Wiretapping Scenario}
\begin{figure}[t!]
    \begin{center}
        \includegraphics[width=2.93in,  height=2.2in]{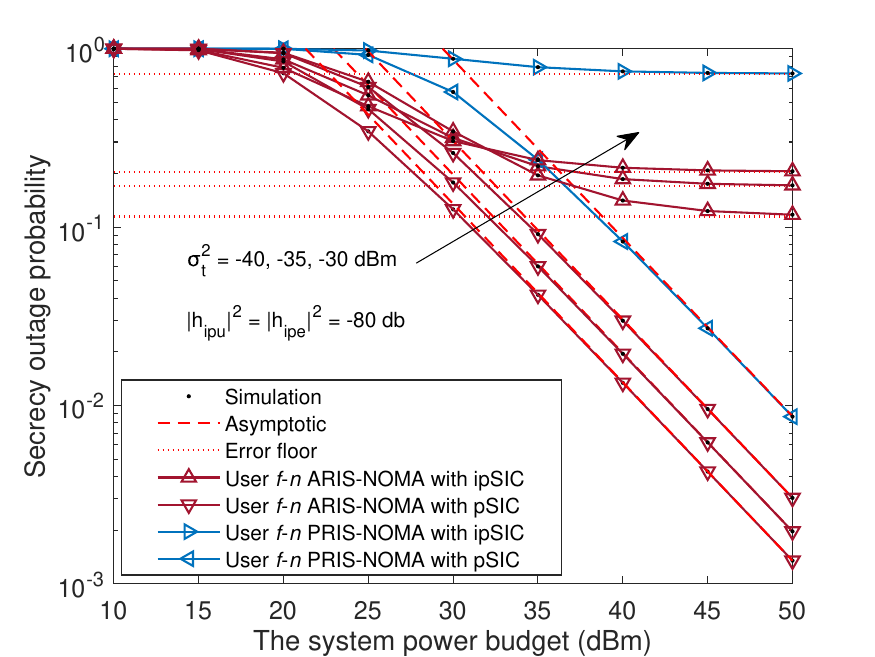}
        \caption{The SOP versus system power budget under internal eavesdropping scenario, with \emph{M} = 40, \emph{P} = 2, \emph{Q} = 20 and ${\left| {{h_{ipu}}} \right|^2} = {\left| {{h_{ipe}}} \right|^2} = $ -80 dB.}
        \label{IE_SOP}
    \end{center}
\end{figure}
Fig. \ref{IE_SOP} plots the SOP versus system power budget under internal eavesdropping scenario. The theoretical analysis curves for user \emph{n} with ipSIC and pSIC are plotted based on (\ref{SOP f n ipSIC}) and (\ref{SOP f n pSIC}), respectively, which match with the simulation results. The effectiveness of SOP expressions are also verified by the asymptotic lines. It can be seen from the figure that the secrecy performance of ARIS-NOMA networks is superior to that of PRIS-NOMA networks. The reason is that active elements can uniquely amplify the power of incident signals and thus enhance the secure transmission. One can also observe that the SOP deteriorates gradually with the increase of thermal noise generated by the active components. This can be explained that larger value of $\sigma _t^2$ impairs the received SNR at user \emph{n}, which can impose a negative effect on the secrecy performance of ARIS-NOMA networks.

\begin{figure}[t!]
    \begin{center}
        \includegraphics[width=2.93in,  height=2.2in]{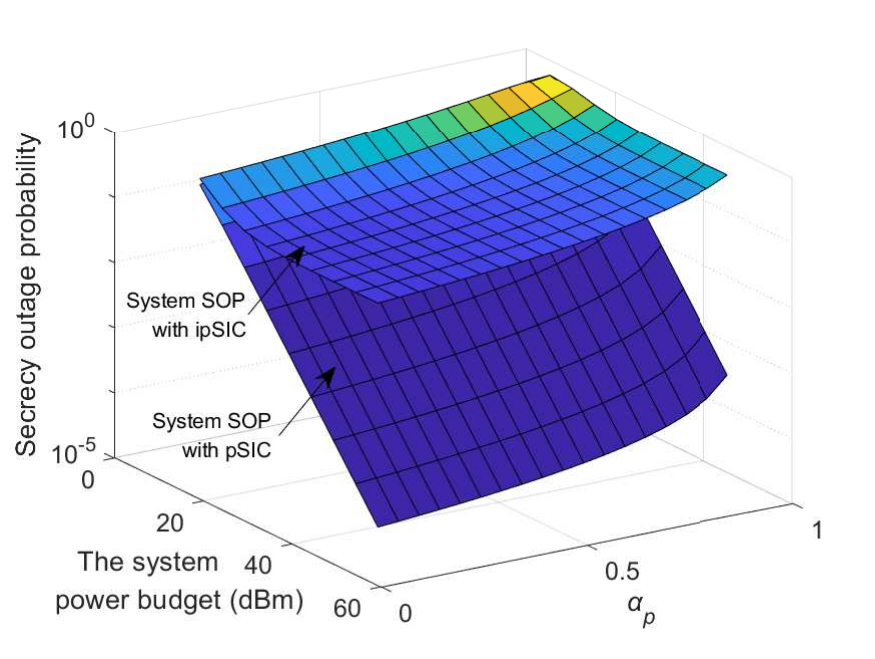}
        \caption{The SOP versus system power budget under internal eavesdropping scenario, with \emph{M} = 40, \emph{P} = 2, \emph{Q} = 20 and ${\left| {{h_{ipu}}} \right|^2} = {\left| {{h_{ipe}}} \right|^2} = $ -70 dB.}
        \label{SOP_diff_powerallocation_IE}
    \end{center}
\end{figure}
Fig. \ref{SOP_diff_powerallocation_IE} plots the SOP versus system power budget under internal eavesdropping scenario. It is can be seen from the figure that the system secrecy outage behaviours are becoming better with the decrease of ${\alpha _P}$, which is the direct opposite of the observation in Fig. \ref{SOP_diff_powerallocation_EE}. The reason is that user \emph{f} with poor channel condition is regarded as the internal Eve and a lower power allocation coefficient can impair the quality of received signals and reduce the wiretapping abilities.

\begin{figure}[t!]
    \begin{center}
        \includegraphics[width=2.93in,  height=2.2in]{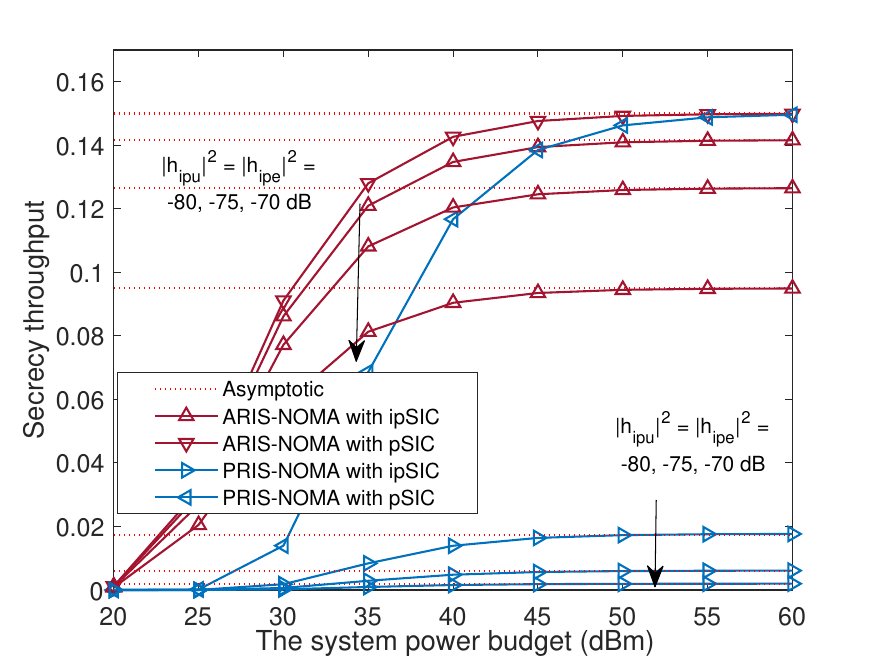}
        \caption{The system secrecy throughput versus system power budget under internal eavesdropping scenario, with \emph{M} = 40, \emph{P} = 2, \emph{Q} = 20, $R_{n}$ = 0.15 BPCU and ${\left| {{h_{ipu}}} \right|^2} = {\left| {{h_{ipe}}} \right|^2} = $ -70 dB.}
        \label{SST_IE}
    \end{center}
\end{figure}
Fig. \ref{SST_IE} plots the system secrecy throughput versus system power budget under internal eavesdropping scenario. The curves for ARIS-NOMA networks with ipSIC and pSIC can be obtained according to (\ref{SOP f n ipSIC}), (\ref{SOP f n pSIC}) and (\ref{SST-IE}), respectively. One can make the following observation from figure that ARIS-NOMA networks can always achieve superior secrecy throughput performance to PRIS-NOMA networks with arbitrary system power budget. The reason is that the additional amplification brought by active components significantly reduces the SOP as displayed in Fig. \ref{IE_SOP}, which makes it possible for ARIS-NOMA to reach higher secrecy throughput. Furthermore, we can observe that the secrecy throughput performance becomes worse with the advancement in residential interference. This phenomenon indicates that more precise hardware design should be implemented at transceivers since the grim influence of ipSIC can bring seriously impairment to the secrecy throughput performance under internal eavesdropping scenario.

\section{Conclusion}\label{Conclusion}\label{SectionV}
In this paper, the physical layer security of ARIS-NOMA networks has been discussed exhaustively by taking into consideration external and internal Eves. The typical on-off control scheme has been utilized to randomly design the phase shifts of ARIS. We have derived the closed-form expressions of SOP and secrecy throughput for ARIS-NOMA networks by exploiting ipSIC and pSIC. To reap more insights, the secrecy diversity orders of legitimate users have been obtained within high SNR region. Numerical results indicate that the secrecy outage performance of ARIS-NOMA networks exceeds that of PRIS-NOMA networks, ARIS/PRIS-OMA networks and several conventional relaying schemes, i.e., AF, HD-DF and FD-DF relay transmissions with the same system power budget. Furthermore, the setup of the system power budget and amplification factor can impose a remarkable influence on the secrecy performance of throughput.

\appendices
\section*{Appendix~A: Proof of Lemma \ref{Lemma1}} \label{AppendixA}
\renewcommand{\theequation}{A.\arabic{equation}}
\setcounter{equation}{0}
We suppose that $Z = {\left| {{\mathbf{v}}_p^H{{\mathbf{D}}_{rn}}{{\mathbf{h}}_{br}}} \right|^2}$ and the CDF of the SINR for user \emph{n} to decode its own signal with ipSIC can be given by
\begin{align}\label{a1}
{F_{\gamma _n^{ipSIC}}}\left( x \right) =& {\rm{Pr}}\left( {\frac{{{c_n}{Z}}}{{{v_n} + {\varpi}P_{BS}^{act}{{\left| {{h_{ipu}}} \right|}^2}}} < x} \right)\notag \\ =& \int_0^\infty  {{f_{{{\left| {{h_{ipu}}} \right|}^2}}}\left( y \right){F_{{Z}}}\left[ {\frac{{x\left( {{v_n} + {\varpi}P_{BS}^{act}y} \right)}}{{{c_n}}}} \right]} dy,
\end{align}where ${c_n} = {a_n}P_{BS}^{act}{\kappa ^2}$, ${v_n} = {\kappa ^2}\sigma _t^2Q{\Omega _{rn}} + {\sigma ^2}$ and ${f_{{{\left| {{h_{ipu}}} \right|}^2}}}\left( y \right) = \frac{1}{{{\Omega _{ipu}}}}{e^{ - y/{\Omega _{ipu}}}}$. Considering that each reconfigurable element in ARIS can regulate the incident signal independently, the channel variables $h_{rn}^m$ with ${m = 1, \cdots ,m, \cdots ,M}$ are not correlated with each other. According to \cite{CunhuaPanARIS}, ${{{\left\| {{\mathbf{h}}_{rn}^H{\mathbf{\Phi }}} \right\|}^2}}$ can be approximated as $Q{\Omega _{rn}}$.
 The PDF of Rayleigh cascaded channel from BS to ARIS and then to user \emph{n} can be acquired based on \cite{liu2014outage}

\begin{align}\label{a2}
{f_Z}\left( z \right) = \frac{2}{{\Gamma \left( Q \right)}}\sqrt {\frac{{{z^{Q - 1}}}}{{{{\left( {{\Omega _{br}}{\Omega _{rn}}} \right)}^{Q + 1}}}}} {K_{Q - 1}}\left( {2\sqrt {\frac{z}{{{\Omega _{br}}{\Omega _{rn}}}}} } \right).
\end{align}By applying integration operation, the CDF of Rayleigh cascaded channel can be written as
\begin{align}\label{a3}
{F_Z}\left( z \right) =& \frac{2}{{\Gamma \left( Q \right){{\left( {\sqrt {{\Omega _{br}}{\Omega _{rn}}} } \right)}^{Q + 1}}}}\notag \\ &\times\int_0^z {{x^{\frac{{Q - 1}}{2}}}{K_{Q - 1}}\left( {2\sqrt {\frac{x}{{{\Omega _{br}}{\Omega _{rn}}}}} } \right)} dx.
\end{align}Replacing \emph{x} with \emph{zy}, (\ref{a3}) can be recast as
\begin{align}\label{a4}
{F_Z}\left( z \right) =& \frac{4}{{\Gamma \left( Q \right)}}{\left( {\sqrt {\frac{z}{{{\Omega _{br}}{\Omega _{rn}}}}} } \right)^{Q + 1}}\notag \\ &\times\int_0^1 {{y^Q}{K_{Q - 1}}\left( {2\sqrt {\frac{z}{{{\Omega _{br}}{\Omega _{rn}}}}} y} \right)dy}\notag \\ \mathop  = \limits^{\left( a \right)}& 1 - \frac{2}{{\Gamma \left( Q \right)}}{\left( {\frac{z}{{{\Omega _{br}}{\Omega _{rn}}}}} \right)^{\frac{Q}{2}}}{K_Q}\left( {2\sqrt {\frac{z}{{{\Omega _{br}}{\Omega _{rn}}}}} } \right),
\end{align}where ${\left( a \right)}$ refers to \cite[Eq. (6.561.8)]{gradvstejn2000table}.
Upon substituting (\ref{a4}) into (\ref{a1}), the CDF of $\gamma _n^{ipSIC}$ can be given by
\begin{small}
\begin{align}\label{a5}
{F_{\hat \gamma _n^{ipSIC}}}\left( x \right) =& \int_0^\infty  {\frac{{e^{ - \frac{y}{{{\Omega _{ipu}}}}}}}{{{\Omega _{ipu}}}}} \left\{ {1 - \frac{2}{{\Gamma \left( Q \right)}}} \right.{\left[ {\frac{{x\left( {{v_n} + \varpi P_{BS}^{act}y} \right)}}{{{c_n}{\Omega _{br}}{\Omega _{rn}}}}} \right]^{\frac{Q}{2}}}\notag \\ &\times \left. {{K_Q}\left[ {2\sqrt {\frac{{x\left( {{v_n} + \varpi P_{BS}^{act}y} \right)}}{{{c_n}{\Omega _{br}}{\Omega _{rn}}}}} } \right]} \right\}dy.
\end{align}
\end{small}The integration in (\ref{a5}) can be solved by introducing Gauss-Laguerre quadrature formula. Specially, we set $t = y/{\Omega _{ipu}}$ and use several straightforward computations, the CDF ${F_{\gamma _n^{ipSIC}}}\left( x \right)$ can be transformed into
\begin{align}\label{a6}
{F_{\gamma _n^{ipSIC}}}\left( x \right) \approx 1 - \frac{2}{{\Gamma \left( Q \right)}}\sum\limits_{d = 1}^D {{G_d}{{\left( {{\Xi _n}x} \right)}^{\frac{Q}{2}}}} {K_Q}\left( {2\sqrt {{\Xi _n}x} } \right).
\end{align}
The proof is completed.

\section*{Appendix~B: Proof of Lemma \ref{Lemma2} } \label{AppendixB}
\renewcommand{\theequation}{B.\arabic{equation}}
\setcounter{equation}{0}
We assume that $Z = {\left| {{\mathbf{v}}_p^H{{\mathbf{D}}_{re}}{{\mathbf{h}}_{br}}} \right|^2}$, the PDF of cascaded eavesdropping channels can be shown as follows according to (\ref{a2})
\begin{small}
\begin{align}\label{b1}
{f_{Z}}\left( z \right) = \frac{{2{z^{\frac{{Q - 1}}{2}}}}}{{\Gamma \left( Q \right){{\left( {\sqrt {{\Omega _{br}}{\Omega _{re}}} } \right)}^{Q + 1}}}}{K_{Q - 1}}\left( {2\sqrt {\frac{z}{{{\Omega _{br}}{\Omega _{re}}}}} } \right).
\end{align}
\end{small}
Therefore, the expectation value of \emph{Z} can be given by
\begin{small}
\begin{align}\label{b2}
\mathbb{E}\left( {Z} \right) &= \int_0^\infty  {z{f_{Z}}\left( z \right)dz}\notag \\  &= \frac{{2{{\left[ {\Gamma \left( Q \right)} \right]}^{ - 1}}}}{{{{\left( {\sqrt {{\Omega _{br}}{\Omega _{re}}} } \right)}^{Q + 1}}}}\int_0^\infty  {{z^{\frac{{Q + 1}}{2}}}} {K_{Q - 1}}\left( {2\sqrt {\frac{z}{{{\Omega _{br}}{\Omega _{re}}}}} } \right)dz.
\end{align}
\end{small}
Upon introducing $z = {x^2}$, (\ref{b2}) can be transformed into
\begin{small}
\begin{align}\label{b3}
\mathbb{E}\left( {Z} \right) = \frac{{4{{\left[ {\Gamma \left( Q \right)} \right]}^{ - 1}}}}{{{{\left( {\sqrt {{\Omega _{br}}{\Omega _{re}}} } \right)}^{Q + 1}}}}\int_0^\infty  {{x^{Q + 2}}}  {K_{Q - 1}}\left( {2x\sqrt {\frac{1}{{{\Omega _{br}}{\Omega _{re}}}}} } \right)dx.
\end{align}
\end{small}
Referring to \cite[Eq. (6.561.16)]{gradvstejn2000table}, the integral in (\ref{b3}) can be written as
\begin{small}
\begin{align}\label{b4}
\int_0^\infty  {{x^\mu }{K_v}\left( {ax} \right)dx}  = {2^{\mu  - 1}}{a^{ - \mu  - 1}}  \Gamma \left( {\frac{{1 + \mu  + v}}{2}} \right)\Gamma \left( {\frac{{1 + \mu  - v}}{2}} \right).
\end{align}
\end{small}
Upon substituting (\ref{b4}) into (\ref{b3}), we can are able to approximately represent ${{\left| {{\mathbf{v}}_p^H{{\mathbf{D}}_{re}}{{\mathbf{h}}_{br}}} \right|}^2}$ as $\mathbb{E}\left( {Z} \right) = Q{\Omega _{br}}{\Omega _{re}}$. The approximate accuracy has been verified by Monte Carlo simulation.

The SOP for external Eve to wiretap user \emph{n}'s information with ipSIC is expressed as follows referring to (\ref{SINR usern hat}), (\ref{the SINR EE n new}) and (\ref{SOP1}).
\begin{small}
\begin{align}\label{b5}
\begin{gathered}
  P_n^{ipSIC}\left( {{R_n}} \right) = {\rm{Pr}}\left[ {\hat \gamma _n^{ipSIC} < {2^{{R_n}}}\left( {1 + \hat \gamma _{e \to n}^{ipSIC}} \right) - 1} \right] \hfill \\
   \approx \int_0^\infty  {{f_{{{\left| {{h_{ipe}}} \right|}^2}}}\left( y \right)} {F_{\hat \gamma _n^{ipSIC}}}\left[ {{2^{{R_n}}}\left( {1 + \Upsilon } \right) - 1} \right]dy, \hfill \\
\end{gathered}
\end{align}
\end{small}where $\Upsilon  = \frac{{{a_n}{\rho _e}{\kappa ^2}Q{\Omega _{br}}{\Omega _{re}}}}{{\left( {{\kappa ^2}\sigma _t^2Q{\Omega _{re}}} \right)/\sigma _e^2 + {\varpi }{\rho _e}y + 1}}$.

Given that ${f_{{{\left| {{h_{ipe}}} \right|}^2}}}\left( y \right) = \frac{1}{{{\Omega _{ipe}}}}{e^{ - \frac{y}{{{\Omega _{ipe}}}}}}$, (\ref{b5}) can be recast as
\begin{small}
\begin{align}\label{b6}
P_n^{ipSIC}\left( {{R_n}} \right) = \frac{1}{{{\Omega _{ipe}}}}\int_0^\infty  {{e^{ - \frac{y}{{{\Omega _{ipe}}}}}}} {F_{\hat \gamma _n^{ipSIC}}}\left[ {{2^{{R_n}}}\left( {1 + \Upsilon } \right) - 1} \right]dy.
\end{align}
\end{small}After replacing $y/{\Omega _{ipe}}$ with \emph{t}, the integration can be solved with the help of Gauss-Laguerre quadrature formula and Lemma \ref{Lemma1} which completes the proof.

\bibliographystyle{IEEEtran}
\bibliography{mybib}

\end{document}